\documentclass[letterpaper,aps,pra,twocolumn,tightenlines,showpacs]{revtex4-1}

\usepackage[T1]{fontenc}
\usepackage{graphicx,color}
\usepackage{graphicx}
\usepackage{subfigure}
\usepackage{latexsym,amsmath,amssymb,amsfonts,amsthm,mathrsfs,bm}
\usepackage{ytableau,youngtab,moresize}
\usepackage{verbatim}
\usepackage[titletoc,title]{appendix}

\newtheorem{thm}{Theorem}[section]
\newtheorem{defn}[thm]{Definition}
\newtheorem{coro}[thm]{Corollary}
\newtheorem{prop}[thm]{Proposition}

\newtheorem*{thmwj}{Theorem III.1}
\newtheorem*{thmij}{Theorem III.2}
\newtheorem*{thm1n}{Theorem III.6}
\newtheorem*{corowj}{Corollary III.3}
\newtheorem*{coroij}{Corollary III.4}
\newtheorem*{corowi}{Corollary III.5}

\newcommand{\tr}[2]{\text{Tr}_{#1}\left(#2\right)}
\newcommand{\trb}[2]{\text{Tr}_{#1}\left[#2\right]}
 
\newcommand{\ket}[1]{| #1 \rangle}
\newcommand{\bra}[1]{\langle #1 |}

\newcommand{\ketbra}[1]{\ket{#1}\bra{#1}} 
\newcommand{\sqket}[1]{|\!#1 \rangle} 
\newcommand{\sqbra}[1]{\langle #1\!|}
\newcommand{\sqketbra}[1]{\sqket{#1}\sqbra{#1}}

\newcommand{\exval}[2]{\langle #1 \rangle_{#2}}

\newcommand{\identity}{\mathbb{I}}
\newcommand{\hilbert}{\mathcal{H}}
\newcommand{\map}{\mathcal{M}}
\newcommand{\neigh}{\mathcal{N}}
\newcommand{\prob}[1]{\text{p}\!\left(#1\right)}

\begin{document}

\title{Compatible quantum correlations: \\ 
on extension problems for Werner and isotropic states}

\author{Peter D. Johnson}
\affiliation{\mbox{Department of Physics and Astronomy, Dartmouth 
College, 6127 Wilder Laboratory, Hanover, NH 03755, USA}}

\author{Lorenza Viola}
\affiliation{\mbox{Department of Physics and Astronomy, Dartmouth 
College, 6127 Wilder Laboratory, Hanover, NH 03755, USA}}

\date{\today}

\begin{abstract}
We investigate some basic scenarios in which a given set of bipartite quantum states may consistently arise as the set of reduced states of a global $N$-partite quantum state. Intuitively, we say that the multipartite state ``joins'' the underlying correlations. Determining whether, for a given set of states and a given joining structure, a compatible 
$N$-partite quantum state exists is known as the quantum marginal problem. We restrict to bipartite reduced states that belong to the paradigmatic classes of Werner and isotropic states 
in $d$ dimensions, and focus on two specific versions of the quantum marginal problem which we find to be tractable. The first is Alice-Bob, Alice-Charlie joining, with both pairs being in a Werner or isotropic state. The second is $m$-$n$ sharability of a Werner state across $N$ subsystems, which may be seen as a variant of the $N$-representability problem to the case where subsystems are partitioned into two groupings of $m$ and $n$ parties, respectively. 
By exploiting the symmetry properties that each class of states enjoys, we determine necessary and sufficient conditions for three-party joinability and 1-$n$ sharability for arbitrary $d$. Our results explicitly show that although entanglement is required for sharing limitations to emerge, correlations beyond entanglement generally suffice to restrict joinability, and not all unentangled states necessarily obey the same limitations. The relationship between joinability 
and quantum cloning as well as implications for the joinability of arbitrary bipartite states are discussed.
\end{abstract}


\date{\today}
\maketitle 

\section{Introduction}
\label{sec:intro}

Understanding the nature of quantum correlations in multiparty systems and the distinguishing features they exhibit relative to classical correlations is a central goal across quantum information processing (QIP) science \cite{Nielsen2001}, with implications ranging from condensed-matter and statistical physics to quantum chemistry, and the quantum-to-classical transition. From a foundational perspective, exploring what different kinds of correlations are, in principle, allowed by probabilistic theories more general than quantum mechanics further helps to identify under which set of physical constraints the standard quantum framework may be uniquely recovered \cite{Masanes2006,Seevinck2010}.

In this context, entanglement provides a distinctively quantum type of correlation, that has no analogue in classical statistical mechanics. A striking feature of entanglement is that it cannot be freely distributed among different parties: if a bipartite system, say, $A$(lice) and $B$(ob), is in a maximally entangled pure state, then no other system, $C$(harlie), may be correlated with it. In other words, the entanglement between $A$ and $B$ is \emph{monogamous} and cannot be \emph{shared} \cite{Wootters2000,Ben,Terhal2004,Verstraete2006,Sanders2011b}. This simple tripartite setting motivates two simple questions about bipartite quantum states: given a bipartite state, we ask whether it can arise as the reduced state of $A$-$B$ and of $A$-$C$ simultaneously; or, more generally, given two bipartite states, we ask if one can arise as the reduced state of $A$-$B$ while the other arises as the reduced state of $A$-$C$. It should be emphasized that both of these are questions about the existence of 
tripartite states with given reduction properties. While formal (and more general) definitions will be provided later in the paper, these examples serve to introduce the notions of \emph{sharing} (1-2 sharing) and \emph{joining} (1-2 joining), respectively. 
In its most general formulation, the joinability problem is also known as the {\em quantum marginal problem} (or {\em local consistency problem}), which has been heavily investigated both from a mathematical-physics \cite{Klyachko,Walter2012,Ruskai} and a quantum-chemistry perspective \cite{Mazziotti2012,Mazziotti2012Rev} and is known to be QMA-hard \cite{Liu2007}. Our choice of terminology, however, facilitates a uniform language for describing the joinability/sharability scenarios. For instance, we say that the joinable correlations of $A$-$B$ and $A$-$C$ are joined by a joining state on $A$-$B$-$C$.

The limited sharability/joinability of entanglement was first quantified in the seminal work by Coffman, Kundu, and Wootters, in terms of an exact (CKW) inequality obeyed by the entanglement across the $A$-$B$, $A$-$C$ and $A$-$(BC)$ bipartitions, as measured by concurrence 
\cite{Wootters2000}. In a similar venue, several subsequent investigations attempted to determine how different entanglement measures can be used to diagnose failures of joinability, see e.g. \cite{Dur2001,Wootters2001,Verstraete2006}. More recently, significant progress has been made in characterizing quantum correlations more general than entanglement \cite{Tufarelli,Modi}, in particular as captured by quantum discord \cite{Ollivier2001}. While it is now established that quantum discord does not obey a monogamy inequality \cite{Adesso2012}, different kind of limitations exist on the extent to which it can be freely shared and/or communicated \cite{Fan2013,Zurek2013}. Despite these important advances, a complete picture is far from being reached. What kind of limitations do strictly mark the quantum-classical correlation boundary? What different quantum features are responsible for enforcing different aspects of such limitations, and how does this relate to the degree of resourcefulness that 
these correlations can have for QIP?

While the above are some of the broad questions motivating this work, our specific focus here is to make progress on joinability and sharability properties in low-dimensional multipartite settings.
In this context, a recent paper \cite{Lieb2013} has obtained a necessary condition for three-party joining in finite dimension in terms of the subsystem entropies, and additionally established a sufficient condition in terms of the trace-norm distances between the states in question and known joinable states. For the specific case of qubit Werner states \cite{Werner1989}, Werner himself established  necessary and sufficient conditions for the 1-2 joining scenario \cite{Werner1990}. With regards to sharability, necessary and sufficient conditions have been found for 1-2 sharing of generic bipartite qubit states \cite{Ranade2009,WalterComm}, as well as for specific classes of qudit states \cite{Myhr2009}. To the best of our knowledge, no conditions that are both necessary and sufficient for the joinability of generic states are available as yet. In this paper, we obtain necessary and sufficient conditions for both the {\em three-party joinability and the 1-$n$ sharability} problems, in the case that the 
reduced bipartite states are {\em either Werner or isotropic states on $d$-dimensional subsystems} (qudits).

Though our results are restricted in scope of applicability, they provide key insights as to the sources of joinability limitations. Most importantly, we find that standard measures of quantum correlations, such as concurrence and quantum discord, do not suffice to determine the limitations in joining quantum correlations. Specifically, we find that the joined states need \emph{not} be entangled or even discordant in order not to be joinable. Further to that, although separable states may have joinability limitations, they \emph{are}, nonetheless, freely (arbitrarily) sharable. By introducing a one-parameter class of probability distributions, we provide a natural classical analogue to qudit Werner and isotropic quantum states. This allows us to illustrate how classical joinability restrictions carry over to the quantum case and, more interestingly, to demonstrate that the quantum case demands limitations which are not present classically. Ultimately, this feature may be traced back to {\em complementarity 
of observables}, which clearly 
plays no role in the classical case. It is suggestive to note that the uncertainty principle was also shown to be instrumental in constraining the sharability of quantum discord \cite{Fan2013}. It is our hope that further pursuits of more general necessary and
sufficient conditions may be aided by the methods and findings herein.

The content is organized as follows. In Sec. \ref{sec:basics} we present the relevant mathematical framework for defining the joinability and sharability notions and the extension problems of interest, along with some preliminary results contrasting the classical and quantum cases. Sec. \ref{sec:Wernerstates} contains the core results of our analysis. In particular, 
after reviewing the defining properties of Werner and isotropic states on qudits, in Sec. \ref{sec:states} we motivate the appropriate choice of probability distributions that serve 
as a classical analogue, and determine the resulting classical joinability limitations in Sec. \ref{sec:classicaljoining}. 
Necessary and sufficient conditions for three-party joinability 
of quantum Werner and isotropic states are established in Sec. \ref{subsec:Wjoin}, and 
contrasted to the classical scenario. Sec. \ref{subsec:cloning} shows how 
the results on isotropic state joinability are in fact related to known results 
on quantum cloning, whereas in Sec. \ref{subsec:Wsharability} we establish simple analytic 
expressions for the 1-$n$ sharability of both Werner and isotropic states, 
along with discussing constructive procedures to determine $m$-$n$ sharability properties 
for $m>1$. In Sec. \ref{sec:furtherremarks}, we present additional remarks on 
joinability and sharability scenarios beyond those of Sec. \ref{sec:Wernerstates}. 
In particular, we outline generalizations of our analysis to 
$N$-party joinability, and show how bounds on the sharability of 
{\em arbitrary bipartite states} follow from the Werner and isotropic results. 
Concluding remarks and open questions are presented in Sec. \ref{sec:end}. For ease and 
clarity of presentation, the technical proofs of the results in Sec. \ref{sec:Wernerstates} 
are presented in two separate Appendixes (\ref{sec:joiningreptools} on joinability 
and \ref{sec:sharingreptools} on sharability, respectively), together with 
the relevant group-representation tools.

\section{Joining and sharing classical vs. quantum states}
\label{sec:basics}

Although our main focus will be to quantitatively characterize simple
low-dimensional settings, we introduce the relevant concepts with a
higher degree of generality, in order to better highlight the
underlying mathematical structure and to ease connections with
existing related notions in the literature.
We are interested in the correlations among the subsystems of a
$N$-partite composite system $S$. In the quantum case, we thus 
require a Hilbert space with a tensor product structure:
\begin{equation*}
\hilbert^{(N)} \simeq \bigotimes_{i=1}^N \hilbert^{(1)}_i, \;\;\;
\text{dim}(\hilbert^{(1)}_i)\equiv d_i,
\end{equation*}
where $\hilbert^{(1)}_i$ represents the individual ``single-particle''
state spaces and, for our purposes, each $d_i$ is finite. In the
classical scenario, to each subsystem we associate a sample space
$\Omega_i$ consisting of $d_i$ possible outcomes, with the joint
sample space being given by the Cartesian product:
\begin{equation*}
\Omega^{(N)} \simeq \Omega_1\times\ldots\times\Omega_N.
\end{equation*}
Probability distributions on $\Omega^{(N)}$ are the classical
counterpart of quantum density operators on $\hilbert^{(N)}$.

\subsection{Joinability}

The input to a joinability problem is a set of subsystem states which,
in full generality, may be specified relative to a ``neighboorhood
structure'' on $\hilbert^{(N)}$ (or $\Omega^{(N)}$) \cite{Ticozzi2012,Ticozzi2013}.
That is, let neighborhoods $\{ {\cal N}_j \}$ be given as subsets of
the set of indexes labeling individual subsystems, ${\cal
N}_k\subsetneq \mathbb{Z}_N$. We can then give the following:
\begin{defn}
{\bf [Quantum Joinability]} Given a neighborhood structure 
$\{\neigh_1,\neigh_2,\ldots,\neigh_{\ell}\}$ on $\hilbert^{(N)}$, a
list of density operators $(\rho_1,\ldots,\rho_{\ell}) \in
(\mathcal{D}(\hilbert_{\neigh_1}),\ldots,
\mathcal{D}(\hilbert_{\neigh_{\ell}}))$ is \emph{joinable} if there
exists an $N$-partite density operator
$w\in\mathcal{D}(\hilbert^{(N)})$, called a \emph{joining state}, that
reduces according to the neighborhood structure, that is,
\begin{equation}
\tr{\hat{\neigh}_k}{w}=\rho_k, \;\;\;\forall k =1, \ldots, \ell,
\label{QJ}
\end{equation}
where $\hat{\neigh}_k\equiv \mathbb{Z}_N \setminus \neigh_k$ is the
{tensor complement} of ${\cal N}_k$.
\end{defn} 
\noindent 
The analagous definition for classical joinability is obtained by substituting corresponding terms, in particular, by replacing the partial trace over $\hat{\neigh}_k$ with the corresponding marginal probability distribution. As remarked, the question of joinability has been extensively investigated in the context of the classical \cite{Fritz2013} and quantum \cite{Klyachko,Sudbery2006,Hall,Lieb2013} marginal problem. A joining state is equivalenty referred to as an \emph{extension} or an element of the {\em pre-image} of the list under the reduction map, while the members of a list of joinable states are also said to be {\em compatible} or {\em consistent}.

Clearly, a {\em necessary} condition for a list of states to be joinable is that they ``agree'' on any overlapping reduced states. That is, given any two states from the list whose neighborhoods are intersecting, the reduced states of the subsystems in the intersection must coincide. From this point of view, any failure of joinability due to a disagreement of overlappping reduced states is a trivial case of non-compatible $N$-party correlations. We are interested in cases where joinability fails \emph{despite} the agreement on overlapping marginals. This consistency requirement will be satisfied by construction for the Werner and isotropic quantum states we shall consider in Sec. \ref{sec:Wernerstates}. 

One important feature of joinability, which has recently been investigated in \cite{Kribs2012}, is the \emph{convex structure} that both joinable states lists and joining states enjoy. The set of lists of density operators satisfying a given joinability scenario is convex under component-wise combination; this is because the same convex combination of their joining states is a valid joining state for the convex combined list of states. Similarly, the set of joining states for a given list of joined states is convex by the linearity of the partial trace.

As mentioned, one of our goals is to shed light on limitations of quantum vs. classical joinability and the extent to which entanglement may play a role in that respect. 
That quantum states are subject to stricter joinability limitations than classical 
probability distributions do, can be immediately appreciated 
by considering two density operators 
$\rho_{AB}=\ketbra{\Psi_{\cal B}}=\rho_{AC}$, where $\ket{\Psi_{\cal
B}}$ is any maximally entangled Bell pair on two qubits: no three-qubit joining state 
$w_{ABC}$ exists, despite the reduced state on $A$ being manifestly consistent.
In contrast, as shown in \cite{Lieb2013,Fritz2013}, as long as two classical distributions have equal marginal distributions over $A$, $p(A,B)$ and $p(A,C)$ can always be joined. This is evidenced by the construction of the joining state: $w(A,B,C)={\prob{A,B}\prob{A,C}}/{\prob{A}}$. As pointed out in \cite{Lieb2013}, although the above choice is not unique, it is the joining state with maximal entropy and represents an even mixture of all valid joining distributions.

Although any {\em two} consistently-overlapped classical probability distributions may be joined, limitations on joining classical probability distributions {\em do} typically arise in more general joining scenarios. This follows from the fact that any classical probability assignments must be consistent with some convex combination of pure states. Consider, for example, a pairwise neighboorhood structure, with an associated list of states $p(A,B)$, $p(B,C)$, and $p(A,C)$, which have consistent single-subsystem marginals. Clearly, if each subsystem corresponds to a bit, no convex combination of pure states gives rise to a probability distribution $w(A, B,C)$ in which each pair is completely anticorrelated; in other words, ``bits of three can't all disagree''. In Sec. \ref{subsec:Wjoin}, we explicitly compare this particular classical joining scenario to analogous quantum scenarios.

While all the classical joining limitations may be expressed by linear inequalities, the quantum joining limitations are significantly more complicated. The limitations arise from demanding that the joining operator be a valid density operator, namely, trace-one and {\em non-negative} (which clearly implies Hermiticity). This fact is demonstrated by the following proposition, which may be readily generalized to any joining scenario:
\begin{prop}
For any two trace-one Hermitian operators $Q_{AB}$ and $Q_{AC}$ which obey the consistency condition $\tr{B}{Q_{AB}}=\tr{C}{Q_{AC}}$, there exists  a trace-one Hermitian joining operator $Q_{ABC}$.
\end{prop}
 \begin{proof}
Consider an orthogonal Hermitian product basis which includes the identity for each subsystem, that is, $\{A_i\otimes B_j \otimes C_k\}$, where $A_0=B_0=C_0=\mathbb{I}$. Then we can construct the space of all valid joining operators $Q_{ABC}$ as follows. Let $d_{ABC}$ be the dimension of the composite system. The component along $A_0\otimes B_0 \otimes C_0$ is fixed as $1/d_{ABC}$, satisfying the trace-one requirement. The components along the two-body operators of the form $A_i\otimes B_j\otimes \mathbb{I}$ are fixed by the required reduction to $Q_{AB}$, and similarly the components along the two-body operators of the form $A_i\otimes \mathbb{I} \otimes C_k$ are determined by $Q_{AC}$. The components along the one-body operators of the form $A_i\otimes \mathbb{I} \otimes \mathbb{I}$, $ \mathbb{I} \otimes B_i\otimes \mathbb{I}$, and $\mathbb{I} \otimes \mathbb{I} \otimes C_i$ are determined from the reductions of $Q_{AB}$ and $Q_{AC}$. This leaves the coefficients of all remaining basis 
operators 
unconstrained, since their corresponding basis operators are zero after a partial trace over systems $B$ or $C$.
\end{proof}
Thus, requiring the joining operator to be Hermitian and normalized is not a limiting constraint with respect to joinability: {\em any limitations are due to the non-negativity constraint}.
Understanding how non-negativity manifests is extremely difficult in general and far beyond our scope here. We can nevertheless give an example in which the role of non-negativity is clear.
Part of the job of non-negativity is to enforce constraints that are also obeyed by classical probability distributions. For example, in the case of a two-qubit state $\rho$, if $\exval{X\otimes \identity}{\rho}=1$ and $\exval{\identity\otimes X}{\rho}=1$, then $\exval{X\otimes X}{\rho}$ \emph{must} equal $1$. More generally, consider a set of mutually commuting observables $\{M_i\}_{i=1}^k$ and any basis $\{\ket{m}\}$ in which all $M_i$ are diagonal. Any valid state must lead to a list of expectation values $(\tr{}{\rho M_1}, \ldots, \tr{}{\rho M_k})$, whose values are element-wise convex combinations of the vertexes $\{(\bra{m}M_1\ket{m}, \ldots, \bra{m}M_k\ket{m})|\forall m\}$. The interpretation of this constraint is that since commuting observables have simultaneously definable values, just as classical observables do, probability distributions on them must obey the rules of classical probability distributions. We call on this fact when we compare the quantum joining limitations to the classical 
analogue ones in Sec. \ref{subsec:Wjoin}.

Non-negativity constraints that do \emph{not} arise from classical limitations on compatible observables may be labeled as inherently quantum constraints, the most familiar being provided by uncertainty relations for conjugate observables \cite{Deutch1983,Wehner2010}. 
Although complementarity constraints are most evident for observables acting on the same system, complementarity can also give rise to a trade-off in the information about a subsystem observable vs. a joint observable. This fact is essentially what allows Bell's inequality to be violated. For our purposes, the complementarity that comes into play is that between ``overlapping'' joint observables (e.g., between $\vec{S}_1\cdot\vec{S}_2$ and $\vec{S}_1\cdot\vec{S}_3$ for three qubits). We are thus generally interested in understanding the interplay between purely classical and quantum joining limitations, and in the correlation trade-offs that may possibly emerge. 

Historically, as already mentioned, a pioneering exploration of 
the extent to which quantum correlations can be shared among three parties 
was carried out in \cite{Wootters2000}, yielding a characterization of the monogamy of 
entanglement in terms of the well-known CKW inequality:
\begin{equation*}
\label{CKW}
\mathcal{C}^2_{AB}+\mathcal{C}^2_{AC} \leq
(\mathcal{C}^2)^{\text{min}}_{A(BC)}, 
\end{equation*}
where ${\mathcal C}$ denotes the concurrence and the right hand-side
is minimized over all pure-state decompositions. Thus, with the
entanglement across the bipartition $A$ and $(BC)$ held fixed, an
increase in the upper bound of the $A$-$B$ entanglement can only come
at the cost of a decrease in the upper bound of the $A$-$C$
entanglement. One may wonder whether the CKW inequality may help 
in diagnosing joinability of reduced states.
If a joining state $w_{ABC}$ is \emph{not} a priori determined (in fact, the existence 
of such a state \emph{is} the entire question of joinability), the CKW inequality may 
be used to obtain a \emph{necessary} condition for joinability, namely, 
if $\rho_{AB}$ and $\rho_{AC}$ are joinable, then  
\begin{equation}
\label{CKWweak}
\mathcal{C}^2_{AB}+\mathcal{C}^2_{AC} \leq 1 .
\end{equation}
However, there exist pairs of bipartite states -- both unentangled (as the following  
Proposition shows) and non-trivially entangled (as we shall determine in Sec. III.B, see 
in particular Fig. \ref{fig:wernerpairjoinability}) -- that obey the ``weak'' CKW inequality 
in Eq. (\ref{CKWweak}), yet are \emph{not} joinable. The key point is that while the   
limitations that the CKW captures are to be ascribed to entanglement, 
entanglement is not required to prevent two states from being 
joinable. In fact, weaker forms of quantum correlations, as quantified by quantum discord 
\cite{Ollivier2001}, are likewise \emph{not} required for joinability limitations. 
Consider, specifically, so-called ``classical-quantum'' bipartite states, of the form
$$\rho=\sum_i p_i |i\rangle \langle i|_A \otimes \sigma_B^i,\;\;\;
\sum_i p_i=1,$$ 
where $\{|i\rangle_A \}$ is some local orthogonal basis on $A$ and
$\sigma_B^i$ is, for each $i$, an arbitrary state on $B$. Such states are known 
to have zero discord \cite{Dakic}. Yet, the following holds:

\begin{prop}
Classical-quantum correlated states need not be joinable.
\end{prop}
\begin{proof}
Consider the two quantum states
\begin{align*}
& \rho_{AB}=(\sqketbra{\uparrow_X\uparrow_X}+
\sqketbra{\downarrow_X\downarrow_X})/2, \\ & \rho_{AC}
=(\sqketbra{\uparrow_Z\uparrow_Z}+
\sqketbra{\downarrow_Z\downarrow_Z})/2,
\label{CQ}
\end{align*}
on the pairs $A$-$B$ and $A$-$C$, respectively. Both have a completely
mixed reduced state over $A$ and thus it is meaningful to consider
their joinability. 
Let $w_{ABC}$ be a joining state. Then the outcome of Bob's $X$
measurement would correctly lead him to predict Alice to be in the
state $\ket{\uparrow_X}$ or $\ket{\downarrow_X}$, while at the same
time the outcome of Charlie's $Z$ measurement would correctly lead him
to predict Alice to be in the state $\ket{\uparrow_Z}$ or
$\ket{\downarrow_Z}$. Since this violates the uncertainty principle, 
$w_{ABC}$ cannot be a valid joining state.
\end{proof}

The existence of separable not joinable states 
has been independently reported in \cite{Lieb2013}. 
While formally our example is subsumed under the more general one presented 
in Thm. 4.2 therein (strictly satisfying the necessary condition for joinability 
given by their Eq. (2.2)), it has the advantage of offering both a transparent physical 
interpretation of the underlying correlation properties, and an intuitive proof of the joinability failure.

\subsection{Sharability}
\label{sec:S}

As mentioned, the second joinability structure we analyze is motivated
by the concept of sharability. In our context, we can think of
sharability as a restricted joining scenario in which a bipartite
state is joined with copies of itself. If ${\mathcal H}^{(2)} \simeq
\mathcal{H}_1^{(1)}\otimes\mathcal{H}_2^{(1)}$,
consider a $N$-partite space that consists of $m$ ``right'' copies of
$\mathcal{H}_1^{(1)}$ and $n$ ``left'' copies of
$\mathcal{H}_2^{(1)}$, with each neighborhood consisting of one right
and one left subsystem, respectively (hence a total of $mn$
neighborhoods). We then have the following:

\begin{defn} 
{\bf [Quantum Sharability]} A bipartite density operator $\rho \in
{\mathcal D}(\mathcal{H}_L \otimes\mathcal{H}_R)$ is $m$-$n$ {\em
sharable} if there exists an $N$-partite density operator $w \in
{\mathcal D}(\mathcal{H}^{\otimes m}_L \otimes\mathcal{H}^{\otimes
n}_R)$, called a {\em sharing state}, that reduces left-right-pairwise
to $\rho$, that is,
\begin{equation}
\tr{\hat{L_i}\hat{R_j}}{w}=\rho, \;\;\; \forall i =1,\ldots, m, \,
j=1, \ldots, n,
\label{QS}
\end{equation}
where the partial trace is taken over the tensor complement of
neighborhood $ij$.
\end{defn} 
\noindent Each $m$-$n$ sharability scenario may be viewed as a specific joining structure with the additional constraint that each of the joining states be equal to one another, the list being $(\rho, \rho, \ldots, \rho)$.
In what follows, we shall take \emph{arbitrarily sharable} to mean $\infty$-$\infty$ sharable, whereas \emph{finitely sharable} means that $\rho$ is not $m$-$n$ sharable for some $m$, $n$. Also, each property ``$m$-$n$ sharable'' (sometimes also referred to as a ``$m$-$n$ extendible'') is taken to define a \emph{sharability criterion}, which a state may or may not satisfy.

It is worth noting the relationship between sharability and \emph{N-representability}. The $N$-representability problem asks if, for a given (symmetric) $p$-partite density operator $\rho$ on $({\mathcal H}_1^{(1)})^{\otimes p}$, there exists an $N$-partite pre-image state for which $\rho$ is the $p$-particle reduced state. $N$-representability has been extensively studied for indistinguishable bosonic and fermionic subsystems \cite{Coleman,Mazziotti2012,Mazziotti2012Rev} and is a very important problem in quantum chemistry \cite{Stillinger1995}. We can view $N$-representability as a variant on the sharability problem, whereby the distinction between the left and right subsystems is lifted, and $m+n=N$. Given the $p$-partite state $\rho$ as the shared state, we ask if there exists a sharing $N$-partite state which shares $\rho$ among all possible $p$-partite subsystems. In the setting of indistinguishable particles, the associated symmetry further constrains the space of the valid $N$-partite sharing states. 

Just as with 1-2 joinability, any classical probability distribution is arbitrarily sharable 
\cite{Seevinck2010}. 
Likewise, similar to the joinability case, convexity properties play an
important role towards characterizing sharability. If 
$\text{dim}(\hilbert^{(1)}_1)=d_1\equiv d_L$ and $\text{dim}(\hilbert^{(1)}_2)=d_2\equiv d_R$,  
then it follows from the convexity of the set of joinable states lists that 
$m$-$n$ sharable states form a convex set, for fixed subsystem dimensions 
$d_L$ and $d_R$. This implies that if $\rho$ satisfies a particular
sharability criterion, then any mixture of $\rho$ with the completely
mixed state also satisfies that criterion, since the completely mixed
state is arbitrarily ($\infty$-$\infty$) sharable. 

Besides mixing with the identity, the degree of sharability may be
unchanged under more general transformations on the input state.
Consider, specifically, completely-positive trace-preserving bipartite
maps $\map(\rho)$ that can be written as a mixture of local unitary
operations, that is,
\begin{equation}
\map (\rho)=\sum_i \lambda_i U^{i}_1\otimes V^{i}_2 \rho
{U^{i}_1}^{\dagger}\otimes {V^{i}_2}^{\dagger}, \;\;\sum_i \lambda_i=1,
\label{map}
\end{equation}
where $U^{i}_1$ and $V^{i}_2$ are arbitrary unitary transformations
$\mathcal{H}_L$ and $ \mathcal{H}_R$, respectively. These (unital) maps form a proper 
subset of general Local Operations and Classical Communication (LOCC) \cite{Nielsen2001}. 
We establish the following:
\begin{thm}
If $\rho$ is $m$-$n$ sharable, then $\map(\rho)$ is $m$-$n$ sharable
for any map $\map$ that is a convex mixture of unitaries. 
\label{thm:LOCC}
\end{thm}
\begin{proof}
Let $\map(\rho)$ be expressed as in Eq. (\ref{map}). By virtue of
the convexity of the set of $m$-$n$ sharable states (for fixed subsystem dimensions),
it suffices to show that each term, $UV\rho U^{\dagger}V^{\dagger}$, in 
$\map(\rho)$ is $m$-$n$ sharable.
Let $w$ be a sharing state for $\rho$, and define
\begin{equation*}
w'\hspace*{-1mm}=\hspace*{-1mm}
\big(U_1\ldots \hspace*{-.5mm}U_m V_{m+1}\ldots \hspace*{-.5mm}
V_{m+n}\big) w \big(U_1^\dag\ldots
\hspace*{-0.5mm}U_m^\dag V_{m+1}^\dag \ldots \hspace*{-.5mm}
V_{m+n}^\dag\big).
\end{equation*}
Then, for any left-right pair of subsystems $i$ and $j$, it follows
that
\begin{equation*} 
\tr{i,j}{w'}=U_i V_j \tr{i,j}{w} U_i^\dag V_j^\dag=U\otimes V \rho
U^\dag \otimes V^\dag=\rho_{UV}.
\end{equation*} 
Hence, $w'$ is an $m$-$n$-sharing state for $\rho_{UV}$, as desired. 
\end{proof}

This result suggests a connection between the degree of sharability
and the entanglement of a given state. In both cases, 
there exist classes of states for which these properties cannot be
``further degraded'' by locally acting maps (or any map for that matter). Obviously, LOCC cannot 
decrease the entanglement of states with no entanglement, and convex unitary 
mixtures as above cannot increase the sharability of states with $\infty$-$\infty$ sharability (because they are already as sharable as possible). 
These two classes of states can in fact be shown to coincide as a consequence
of the fact that \emph{arbitrary sharability is equivalent to
(bipartite) separability}. This result has been appreciated in the literature 
\cite{Terhal2004,Masanes2006,Yang2006,Seevinck2010} and is credited to both 
\cite{Werner1989b} and \cite{Fannes1988}. We reproduce it here in view 
of its relevance to our work:
\begin{thm}
A bipartite quantum state $\rho$ on
$\mathcal{H}_L\otimes\mathcal{H}_{R}$ is unentangled (or separable) if
and only if it is arbitrarily sharable.
\end{thm}

\begin{proof} 
($\Leftarrow$) Let $\rho$ be separable. Then for some set of density
operators $\{\rho^L_i,\rho^R_i\}$, it can be written as
$\rho=\sum_i \lambda_i \rho^L_i \otimes \rho^R_i,$
with $\sum_i \lambda_i =1$. Let $n$ and $m$ be arbitrary, and let the
$N$-partite state $w$,  be defined as follows:
\begin{equation*}
w= \sum_i \lambda_i (\rho_i^L)^{\otimes m} \otimes
(\rho_i^R)^{\otimes n},
\end{equation*}
with $N=m+n$. By construction, the state of each $L$-$R$ pair is
$\rho$, since it follows straighforwardly that Eq. (\ref{QS}) is
obeyed for each $i,j$. Thus, $w$ is a valid sharing state.

($\Rightarrow$) Since $\rho$ is arbitrarily sharable, there exists a
sharing state $w$ for arbitrary values of $m$, $n$. In particular, we
need only make use of a sharing state $w$ for $m=1$ and arbitrarily
large $n$, whence we let $n\rightarrow \infty$. Given $w$, let us
construct another sharing state $\tilde{w}$, which is invariant under
permutations of the right subsystems, that is, let
\begin{equation*}
\tilde{w}=\frac{1}{|S_n|}\sum_{\pi \in S_n} V_\pi^\dagger w V_\pi , 
\end{equation*}
where $S_n \equiv \{ \pi\}$ is the permutation group of $n$ objects, acting 
on $\mathcal{H}^{\otimes n}_R$ via the natural $n$-fold representation,  
$V_\pi (\prod_i \ket{\psi_i}) = \otimes_i \ket{\psi_{\pi(i)}}$, $i=1,\ldots , n$.
It then follows that $\tilde{w}$ shares $\rho$:
\begin{align*}
\tr{\hat{L},\hat{R}}{\tilde{w}}&=\frac{1}{|S_n|}\sum_{\pi \in S_n}
    \tr{\hat{L},\hat{R}}{V_\pi^\dagger w V_\pi }\\ 
    &=\frac{1}{|S_n|}\sum_{\pi \in S_n}
    \tr{\hat{L},\pi (\hat{R_i})}{w} 
    =\frac{1}{|S_n|}\sum_{\pi\in S_n} \rho = \rho.
\end{align*}
Having established the existence of a symmetric sharing state
$\tilde{w}\in {\mathcal
D}(\mathcal{H}_L\otimes\mathcal{H}^{\otimes\infty}_R)$, Fannes'
Theorem (see section 2 of \cite{Fannes1988}) implies the existence of a
unique representation of $\tilde{w}$ as a sum of product states, 
$\tilde{w}=\sum_i \lambda_i
\rho_L^i\otimes\rho_R^i\otimes\rho_R^i\otimes \ldots. $
Reducing $\tilde{w}$ to any $L$-$R$ pair leaves a separable
state. Thus, if $\rho$ is 1-$n$ sharable it must be separable.
\end{proof}

\noindent 
As we alluded to before, a Corollary of this result is that
in fact {\em 1-$\infty$ sharability implies} $\infty$-$\infty$
sharability. In closing this section, we also briefly mention the
concept of {\em exchangeability} \cite{Caves2002,Renner2007}. A density operator
$\rho$ on $(\mathcal{H}^{(1)}_1)^{\otimes p}$ is said to be exchangeable if it
is symmetric under permutation of its $p$ subsystems and if there
exists a symmetric state $w$ on $(\mathcal{H}_1^{(1)})^{\otimes (p+q)}$ such
that the reduced states of any subset of $p$ subsystems is $\rho$ for all
$q\in\mathbb{N}$. Similar to sharability, exchangeability implies
separability. However, {\em the converse only holds in general for
sharability}: clearly, there exist states which are separable but not
exchangeable, because of the extra symmetry requirement. Thus, the
notion of sharability is more directly related to entanglement than
exchangeability is.

\section{Joining and Sharing Werner and Isotropic States}
\label{sec:Wernerstates}

Even for the simplest case of two bipartite states with an overlapping marginal, a general 
characterization of joinability is extremely non-trivial. As remarked, no conditions yet exist 
which are both necessary and sufficient for two \emph{arbitrary} density operators to be joinable; 
although, conditions that are separately necessary or sufficient have been recently derived 
\cite{Lieb2013}. In this Section, we present a complete characterization of the three-party 
joining scenario and the 1-$n$ sharability problem for 
Werner and isotropic states on arbitrary subsystem dimension $d$. We begin by introducing the 
relevant families of quantum and classical states to be considered.

\subsection{Werner and isotropic qudit states, and their classical analogues}
\label{sec:states}

The usefulness of bipartite Werner and isotropic states is derived from their simple analytic properties and range of mixed state entanglement. For a given subsystem dimension $d$, Werner states are defined as the one-parameter family that is invariant under collective unitary transformations \cite{Werner1989} (see also \cite{Renner2007}), that is, transformations of the form $U \otimes U$, for arbitrary $U \in {\mathfrak U}(d)$. The parameterization which we employ is given by
$$ 
\rho(\Psi^{-})=\frac{d}{d^2-1}\left[ (d-\Psi^{-}) \frac{\identity}{d^2} 
+ \Big(\Psi^{-}- \frac{1}{d}\Big)\frac{V}{d} \right], $$
where $V$ is the swap operator, defined by its action on any product
ket, $V\ket{\psi \phi} \equiv \ket{\phi \psi}$. This parameterization is chosen because $\Psi^-$ is a Werner state's expectation value with respect to $V$, $\Psi^{-}=\mbox{Tr}[V\rho(\Psi^{-})]$. Non-negativity is ensured by $-1\leq \Psi^{-} \leq 1$, and the completely mixed state corresponds to $\Psi^{-}=1/d$. Furthermore, the concurrence of Werner states is simply given by \cite{Chen}
\begin{equation}
\mathcal{C}(\rho(\Psi^{-}))=-\trb{}{V\rho(\Psi^{-})}=-\Psi^{-} , \;\;\; \Psi^{-} \leq 0. 
\label{Wc}
\end{equation}
For $\Psi^{-} > 0$, the concurrence is defined to be zero, indicating
separability. Werner states have been experimentally characterized
for photonic qubits, see e.g. \cite{KwiatW}. Interestingly, they can
be dissipatively prepared as the steady state of coherently driven
atoms subject to collective spontaneous decay \cite{Agarwal2006}.

Isotropic states are defined, similarly, as the one-parameter family that is invariant under transformations of the form $U^* \otimes U$ \cite{Horodecki1999Jun}. We parameterize 
these states as
$$ \rho(\Phi^{+})=\frac{d}{d^2-1}\left[(d-\Phi^{+})\frac{\identity}{d^2}+
 \Big(\Phi^{+}-\frac{1}{d}\Big)\ketbra{\Phi^+}\right], $$
where $\ket{\Phi^{+}}=\sqrt{1/d}\sum_{i}\ket{ii}$. The value of the parameter is given by the expectation value with respect to the partially transposed swap operator, $\Phi^{+}=
\trb{}{V^{T_A}_{(AB)}\rho(\Phi^{+})}$, and is related to the so-called ``singlet fraction'' \cite{Horodecki1999} by $\Phi^{+}=dF$. Non-negativity is now 
ensured by $0\leq\Phi^+\leq d$, whereas the concurrence is given by \cite{Caves2003},
\begin{align}
\mathcal{C}(\rho(\Phi^{+}))=\sqrt{\frac{2}{d(d-1)}}\,(\Phi^{+}-1) , \;\;\; \Phi^{+} \geq 1,
\end{align}
and is defined to be zero for $\Phi^{+}\leq 1$.

Before introducing probability distributions that will serve as the analogue classical states, 
we present an alternative way to think of Werner states, which will prove useful later. First, 
the highest purity, attained for the $\Psi^{-}=-1$ state, is $2/[d(d-1)]$, with the absolute maximum of 1 corresponding to the pure singlet state for $d=2$. Second, collective projective measurements on a most-entangled Werner state return only disagreeing outcomes (e.g., corresponding to $\ket{1}\otimes\ket{3}$, but not $\ket{1}\otimes\ket{1}$). The following construction of bipartite Werner states demonstrates the origin of both of these essential features. For generic $d$, the analogue to the singlet state is the following $d$-partite fully anti-symmetric state:
\begin{align}
\label{eq:purewerner}
 \ket{\psi^-_d}= \frac{1}{\sqrt{d!}}
 \sum_{\pi\in S_d} \text{sign}(\pi) V_{\pi}\ket{1}\ket{2}\ldots\ket{d}, 
\end{align}
where, as before, $S_d \equiv \{ \pi\}$ denotes the permutation group and $\{\ket{\ell}\}$ is an orthonormal basis on ${\mathcal H}^{(1)}\simeq {\mathbb C}^d$. The above state has the property 
of being ``completely disagreeing'', in the sense that a collective measurement returns outcomes that differ on each qudit with certainty. The most-entangled \emph{bipartite} qudit Werner state is nothing but the two-party reduced state of $\ket{\psi^-_d}$. Thus, we can think of general bipartite qudit Werner states as mixtures of the completely mixed state with the two-party-reduction of $\ket{\psi^-_d}$. The inverse of $2/[d(d-1)]$ (the purity) is precisely 
the number of ways two ``dits'' can disagree. Understanding bipartite Werner states to arise from reduced states of $\ket{\psi^-_d}$ will inform our construction of the classical analogue states,  and also help us understand some of the results of Sec. \ref{subsec:Wjoin} and \ref{subsec:Wsharability}.

For Werner states, increased entanglement corresponds to increased ``disagreement'' for collective measurement outcomes. For isotropic states, increased entanglement corresponds to increased ``agreement'' of collective measurements, but only with respect to the computational basis 
$\{ \ket{i}\}$ relative to which such states are defined. It is this expression of agreement vs. disagreement of outcomes which carries over to the classical analogue states, which we are now ready to introduce.
The relevant probability distributions are defined on the outcome space 
$\Omega_d\times \Omega_d =\{1,\ldots,d\}\times\{1,\ldots,d\}$. To resemble Werner and isotropic quantum states, these probability distributions should have completely mixed marginal distributions and range from maximal disagreement to maximal agreement. This is achieved by an interpolation between an even mixture of ``agreeing pure states'', namely, $(1,1), (2,2), \ldots, (d,d)$, and an even mixture of all possible ``disagreeing pure states'', namely, $(1,2), \ldots, (1,d),(2,1),\ldots,(d,d-1)$. That is:
\begin{equation}
p(A=i,B=j)_{\alpha}=\frac{\alpha}{d}\,\delta_{i,j}+
\frac{1-\alpha}{d(d-1)}\, (1-\delta_{i,j}),
\label{eq:classicalW}
\end{equation}
where $\alpha$ is the probability that the two outcomes agree.

To make the analogy complete, it is desirable to relate $\alpha$ to both $\Psi^-$ and $\Phi^+$. 
We define $\alpha$ in the quantum cases to be the probability of obtaining $\ket{k}$ on system 
$A$, conditional to outcome $\ket{k}$ on system $B$ for the projective measurement $\{\ket{ij}\bra{ij} \}$. For Werner states, this probability is related to $\Psi^{-}$ by
\begin{align}
p(\ket{k}_A\,|\,\ket{k}_B)_W=\frac{\Psi^{-}+1}{d+1} \equiv \alpha_W,
\label{agreewer}
\end{align}
and, similarly for isotropic states, we have
\begin{align}
p(\ket{k}_A\,|\,\ket{k}_B)_I=\frac{\Phi^{+}+1}{d+1} \equiv \alpha_I.
\label{agreeiso}
\end{align}
We may thus re-parameterize both the Werner and isotropic states in terms of their 
respective above-defined ``probabilities of agreement'', namely:
\begin{eqnarray}
\hspace*{-6mm} \rho(\alpha_W)&\hspace*{-0.5mm}=\hspace*{-0.5mm}&\frac{d}{d-1}\left[ (1-\alpha_W)\frac{\identity}{d^2}+
 \Big(\alpha_W-\frac{1}{d}\Big) \frac{V}{d} \right] ,
 \label{eq:agreementparamwer}\\
\hspace*{-6mm} \rho(\alpha_I)&\hspace*{-0.5mm}=\hspace*{-0.5mm}&\frac{d}{d-1}\left[ (1-\alpha_I)\frac{\identity}{d^2}+
\Big( \alpha_I- \frac{1}{d}\Big)\ketbra{\Phi^{+}}\right],
\label{eq:agreementparamiso} 
\end{eqnarray}
subject to the conditions
$$0\leq\alpha_W\leq\frac{2}{d+1},\;\;\; \frac{1}{d+1}\leq\alpha_I\leq 1.$$

For Werner states, $\alpha_W$ can rightly be considered a probability of agreement because it is independent of the choice of local basis vectors in the projective measurement $\{U \otimes U\ket{ij}\bra{ij}U^{\dagger}\otimes U^{\dagger} \}$. For isotropic states, $\alpha_I$ does not have as direct an interpretation. We may nevertheless interpret $\alpha$ as a probability of basis-independent agreement if we pair local basis vectors on $A$ with their complex conjugates on $B$. In other words, $\alpha_I$ can be thought of as the probability of agreement for local projective measurements of the form  $\{U^* \otimes U\ket{ij}\bra{ij}U^{*\dagger}\otimes U^{\dagger} \}$ \cite{Remark}.

\subsection{Classical joinability limitations} 
\label{sec:classicaljoining}

In order to determine the joinability limitations in the classical case, we begin by noting that any (finite-dimensional) classical probability distribution is a unique convex combination of the pure states of the system. In our case, there are five extremal three-party states, for which the two-party marginals are classical analogue states, as defined in Eq. (\ref{eq:classicalW}). These are
\begin{align*}
p(A,B,C\,\text{agree})=&\frac{1}{d}\sum_i (i,i,i),\\
p(A,B\,\text{agree})=&\frac{1}{d(d-1)}\sum_{i\neq j} (i,i,j),\\
p(A,C\,\text{agree})= &\frac{1}{d(d-1)}\sum_{i\neq j} (i,j,i),\\
p(B,C\,\text{agree})=&\frac{1}{d(d-1)}\sum_{i\neq j} (j,i,i),\\
p(\text{all disagree})=&\frac{1}{d(d-1)(d-2)}\sum_{i\neq j\neq k} (i,j,k),
\end{align*}
where $(i,j,k)$ stands for the pure probability distribution $p(A,B,C)=\delta_{A,i}\delta_{B,j}\delta_{C,k}$. The first four of these states are valid for all $d\geq2$ and each corresponds to a vertex of a tetrahedron, as depicted in Fig. \ref{fig:correlationcomparison}(left). The fifth state is only valid for $d\geq3$ and corresponds to the point $(\alpha_{AB},\alpha_{AC},\alpha_{BC})=(0,0,0)$ in Fig. \ref{fig:correlationcomparison}(right). Any valid three-party state for which the two-party marginals are classical analogue states must be a convex combination of the above states. Therefore, the joinable-unjoinable boundary is delimited by the boundary of their convex hull. 
For the $d=2$ case, the inequalities describing these boundaries are explicitly given by the following:
\begin{align*}
p(A,B,C\,\text{agree})\geq 0 & \Rightarrow & \;\;\,\alpha_{AB}+\alpha_{AC}+\alpha_{BC} & \geq 1, \\
p(C \,\text{disagrees})\geq 0 & \Rightarrow & \;\;\,-\alpha_{AB}+\alpha_{AC}+\alpha_{BC} & \leq 1, \\
p(B \,\text{disagrees})\geq 0 & \Rightarrow & \alpha_{AB}-\alpha_{AC}+\alpha_{BC} & \leq 1,\\
p(A \,\text{disagrees})\geq 0 & \Rightarrow & \alpha_{AB}+\alpha_{AC}-\alpha_{BC} & \leq 1, 
\end{align*}
where each inequality arises from requiring that the corresponding extremal state has a non-negative likelihood. In the $d\geq3$ case, the inequality $p(A,B,C\,\text{agree})\geq 0$ is replaced by $\alpha_{AB},\alpha_{AC},\alpha_{BC}\geq 0$.

\begin{figure*}[!t]
\begin{centering}
\includegraphics[width=.85\columnwidth,viewport=20 0 1000 900,clip]{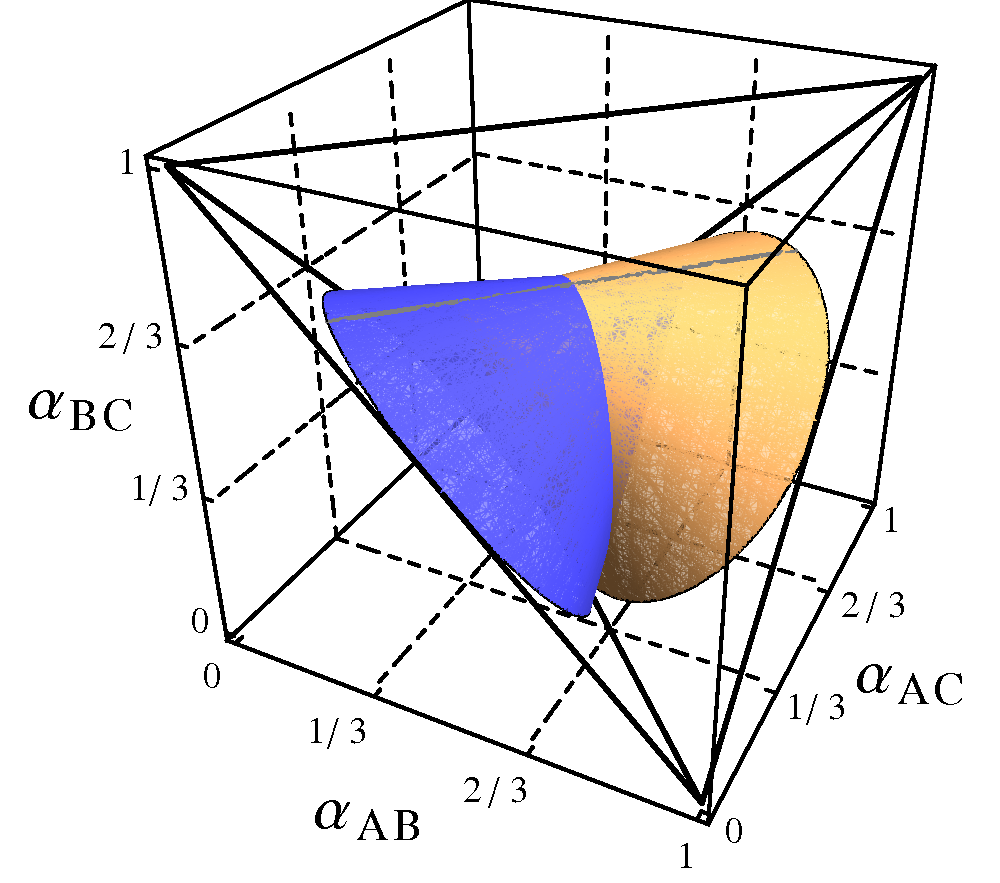}
\hspace{10mm}
\includegraphics[width=.85\columnwidth,viewport=20 0 1000 900,clip]{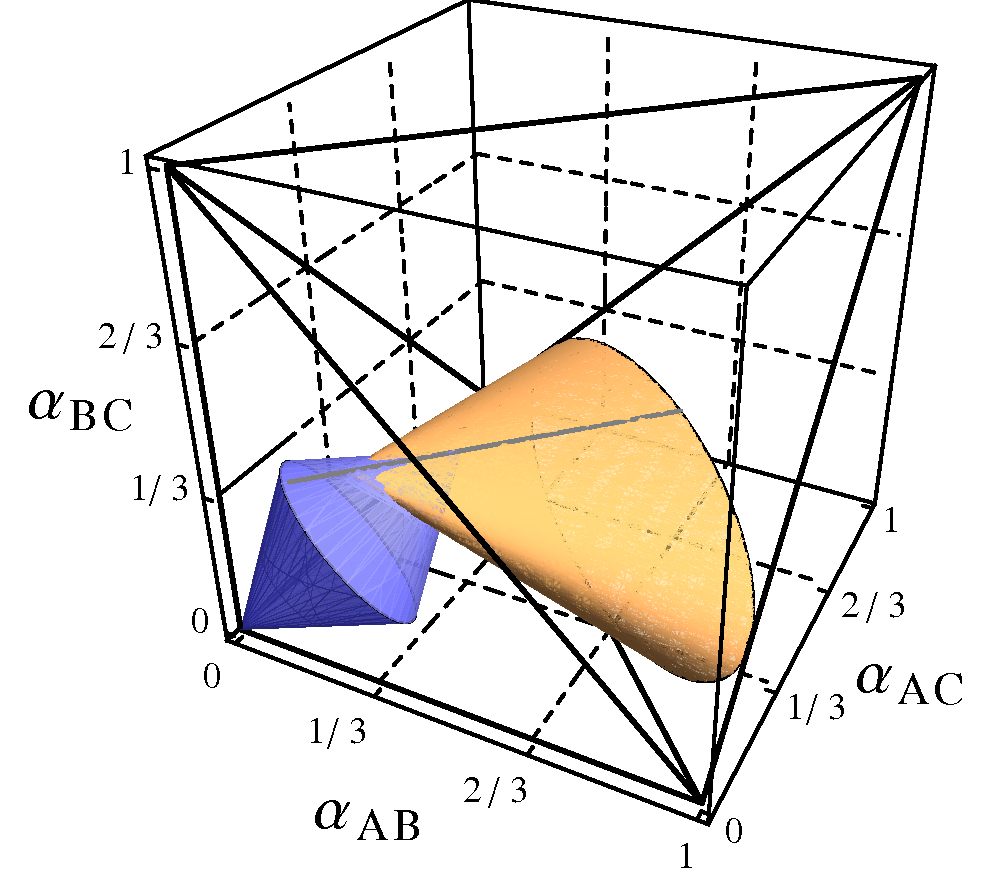}
\end{centering}
\vspace*{-3mm}
\caption{(Color online) 
Three-party quantum and classical joinability limitations for Werner and isotropic states, and their classical analogue, as parameterized by Eqs. (\ref{eq:agreementparamwer}), (\ref{eq:agreementparamiso}), (\ref{eq:classicalW}), respectively. 
Left panel: Qubit case, $d=2$.
The Werner state boundary is the surface of the darker cone with its vertex at $(2/3, 2/3, 2/3)$, whereas the isotropic state boundary is the surface of the lighter cone with its vertex at $(1/3, 1/3, 2/3)$. The classical boundary is the surface of the tetrahedron. 
Right panel: Higher-dimensional case, $d=5$. 
The Werner state boundary is the surface of the bi-cone with vertices at $(0,0,0)$ and $(1/3, 1/3, 1/3)$, whereas the isotropic state boundary is the flattened cone with its vertex at $(1/6, 1/6, 1/3)$. The classical boundary is the surface of the two joined tetrahedra. In both panels the grey line resting on top of the cones indicates the colinearity of the cone surfaces along this line segment.}
\label{fig:correlationcomparison}
\end{figure*}

\subsection{Joinability of Werner and isotropic qudit states}
\label{subsec:Wjoin}

We now present our results on the three-party joinability of Werner and isotropic states 
and then compare them to the classical limitations just found in the previous section. 
While, as mentioned, all the technical proofs are post-poned to Appendix \ref{sec:joiningreptools} in order to ease readability, the basic idea is to exploit the high degree of symmetry that these classes of states enjoy. 

Consider Werner states first. Our starting point is to observe that 
if a tripartite state $w_{ABC}$ joins two reduced Werner states $\rho_{AB}$ and $\rho_{AC}$, 
then the ``twirled state'' $\tilde{w}_{ABC}$, given by
\begin{equation}
\tilde{w}_{ABC}=\int (U\otimes U\otimes U)\, w_{ABC} \,
(U\otimes U \otimes U)^{\dagger} d\mu(U),
\label{twirlapp}
\end{equation}
is also a valid joining state. In Eq. (\ref{twirlapp}), $\mu$ denotes the invariant Haar measure on ${\mathfrak U}(d)$, and the twirling super-operator effects a projection into the subspace of operators with collective unitary invariance \cite{Eggeling2001}. By invoking the Schur-Weyl duality \cite{Fulton1991}, the guaranteed existence of joining states with these symmetries allows one to narrow the search for valid joining states to the Hermitian subspace spanned by representations of subsystem permutations, that is, density operators of the form 
\begin{equation}
w=\sum_{\pi\in{S}_3}\mu_{\pi}V_{\pi},\;\;\; w \in {\mathcal W}_W, 
\label{permrepspan}
\end{equation}
where Hermiticity demands that $\mu_{\pi}^*=\mu_{\pi^{-1}}$.
%
Given $w_{ABC}$ which joins Werner states, each subsystem pair is characterized by the 
expectation value with the respective swap operator, 
$\Psi^{-}_{ij}= \mbox{Tr}[ w_{ABC} (V_{ij}\otimes\identity_{\overline{ij}} )],$
where $i,j\in\{A,B,C\}$ with $i\neq j$. Hence, the task is to determine
for which $(\Psi^{-}_{AB}, \Psi^{-}_{BC}, \Psi^{-}_{AC})$ there exists a density
operator $w_{ABC}$ consistent with the above expectations. Our main result is the following:

\begin{thm} 
\label{triower} 
Three Werner qudit states with parameters $\Psi^{-}_{AB}, \Psi^{-}_{BC}, \Psi^{-}_{AC}$ are
joinable if and only if $(\Psi^{-}_{AB}, \Psi^{-}_{BC}, \Psi^{-}_{AC})$ lies within the bi-cone
described by
\begin{align}
1\pm\overline{\Psi^{-}}\geq \frac{2}{3}\left|\Psi^{-}_{BC}+\omega
\Psi^{-}_{AC}+\omega^2\Psi^{-}_{AB}\right| ,
\label{d1}
\end{align}
for $d\geq3$, or within the cone described by
\begin{align}
1-\overline{\Psi^{-}}\geq \frac{2}{3}\left|\Psi^{-}_{BC}+ \omega
\Psi^{-}_{AC}+\omega^2\Psi^{-}_{AB}\right| ,
\;\;\; \overline{\Psi^{-}}\geq 0 , 
\label{d2b} 
\end{align}
for $d=2$, where 
\begin{equation}
\overline{\Psi^{-}}=\frac{1}{3}(\Psi^{-}_{AB}+\Psi^{-}_{BC}+\Psi^{-}_{AC}),
\;\;\;\omega=e^{i\frac{2\pi}{3}}.
\label{Phibar}
\end{equation}
\end{thm}

Similarly, if a tripartite state $w_{ABC}$ joins isotropic states $\rho_{AB}$ and $\rho_{AC}$, then the ``isotropic-twirled state'' $\tilde{w}_{ABC}$, given by
\begin{equation}
\tilde{w}_{ABC}=\int (U^*\otimes U\otimes U)\, w_{ABC} \,
(U^*\otimes U \otimes U)^{\dagger} d\mu(U),
\label{twirliso}
\end{equation}
is also a valid joining state. A clarification is, however, in order at this point: 
although we have been referring to the isotropic joinability scenario of interest as three-party isotropic state joining, this is somewhat of a misnomer because we effectively consider the pair 
$B$-$C$ to be in a Werner state, as evident from Eq. (\ref{twirliso}).
Compared to Eq. (\ref{permrepspan}), the relevant search space is now  
partially transposed relative to subsystem $A$, that is, 
consisting of density operators of the form
\begin{equation}
\label{eq:ptrep}
w=\sum_{\pi\in{S}_3}\mu_{\pi}V_{\pi}^{T_A}, \;\;\; w \in {\mathcal W}_{\text{iso}}.
\end{equation}
Our main result for three-party joinability of isotropic states is then contained in the following:

\begin{thm} 
\label{trioiso}
Two isotropic qudit states $\rho_{AB}$ and $\rho_{AC}$ and qudit Werner state $\rho_{BC}$ with parameters $\Phi^+_{AB}, \Phi^+_{AC}, \Psi^-_{BC}$ are joinable if and only if $(\Phi^+_{AB}, \Phi^+_{AC}, \Psi^-_{BC})$ lies within the cone described by
\begin{eqnarray}
\Phi^{+}_{AB}+\Phi^{+}_{AC}-\Psi^{-}_{BC} \leq d\,  ,
\label{d1iso} 
\end{eqnarray}
 \vspace*{-5mm}
\begin{eqnarray}
 1+\Phi^+_{AB} & + & \Phi^+_{AC} -  \Psi^-_{BC} \geq 
  \label{d1isoa}  \\
 \bigg| d(\Psi^-_{BC} & - &1) +  \sqrt{\frac{2d}{d-1}}(e^{i\theta}\Phi^+_{AB}+
e^{-i\theta}\Phi^+_{AC})\bigg|, \nonumber 
\end{eqnarray}
 \vspace*{-5mm}
$$ e^{\pm i\theta}= \pm i \sqrt{(d+1)/(2d)}+\sqrt{(d-1)/(2d)}, $$
or, for $d\geq 3$, within the convex hull of the above cone and the point $(0,0,-1)$.
\end{thm}

\noindent 
The results of Theorems \ref{triower} and \ref{trioiso} as well as of Sec. \ref{sec:classicaljoining} are pictorially summarized in Fig. \ref{fig:correlationcomparison}.

We now compare these quantum joinability limitations to the joinability limitations in place for classical analogue states. As described in Sec. \ref{sec:classicaljoining}, the non-negativity of $p(A,B,C\, \text{agree})$ and $p(A \,\text{disagree})$ is enforced by the two 
inequalities $\alpha_{AB}+\alpha_{AC}+\alpha_{BC}\geq 1$ and $-\alpha_{AB}-\alpha_{AC}+\alpha_{BC}\geq 1$, respectively. We expect the same requirement to be enforced by the analogue quantum-measurement statistics. For $d=2$, the bases of the Werner and isotropic joinability-limitation cones are determined by  $\Psi^{-}_{AB}+\Psi^{-}_{AC}+\Psi^{-}_{BC}\geq0$ and $\Phi^{+}_{AB}+\Phi^{+}_{AC}-\Psi^{-}_{BC} \leq 2$, respectively. Writing down each of these parameters in terms of the appropriate probability of agreement $\alpha$, as defined in Eqs. (\ref{agreewer}) and (\ref{agreeiso}), we obtain $\alpha_{AB}+\alpha_{AC}+\alpha_{BC}\geq 1$ and $-\alpha_{AB}-\alpha_{AC}+\alpha_{BC}\geq 1$. Hence, for qubits, part of the quantum joining limitations are indeed derived from the classical joining limitations. This is also illustrated in Fig. \ref{fig:correlationcomparison}(left). Of course, one would not expect the quantum scenario to exhibit violations of the classical joinability restrictions; still, it is 
interesting that states which exhibit manifestly non-classical correlations may nonetheless saturate bounds obtained from purely classical joining limitations.

For $d\geq 3$, the only classical boundary which plays a role is the one which bounds the base of the isotropic joinability-limitation cone: $\Phi^{+}_{AB}+\Phi^{+}_{AC}-\Psi^{-}_{BC} \leq d$. Again, in terms of the agreement parameters, this is (just as for qubits) 
$-\alpha_{AB}-\alpha_{AC}+\alpha_{BC}\geq 1$. In the Werner case, the quantum joinability boundary is not clearly delineated by the classical joining limitations. We can nevertheless make the following observation. By the non-negativity of Werner states, the three-party joinability region in Fig. \ref{fig:correlationcomparison}(right) is required to lie within a cube of side-length 
${2}/(d+1)$ with one corner at $(0,0,0)$. Consider the set of cubes obtained by rotating from this initial cube about an axis through $(0,0,0)$ and $(2/(d+1),2/(d+1),2/(d+1)$. It is a curious fact that the exact quantum Werner joinability region (the bi-cone) is precisely the {\em intersection} of all such cubes.

Another interesting feature is that there exist {\em trios of unentangled Werner states 
which are not joinable}. For example, the point 
$(\Psi^{-}_{AB},\Psi^{-}_{AC},\Psi^{-}_{BC})=(1, 1, 0)$
corresponds to three separable Werner states that are not joinable. This point is of
particular interest because its classical analogue \emph{is}
joinable. Translating $(1,1,0)$ into the agreement-probability
coordinates, $(\alpha_{AB},\alpha_{AC},\alpha_{BC})=(2/3,2/3,1/3)$, we
see that this point is actually on the classical joining limitation
border. Thus, these three separable, correlated
states are not joinable for purely quantum mechanical reasons.
Note that the point $(\alpha_{AB},\alpha_{AC},\alpha_{BC})=(2/3,2/3,1/3)$ 
does correspond to a joinable trio of pairs in the isotropic three-party joining scenario: this point lies at the center of the face of the isotropic joinability cone, as seen in Fig. \ref{fig:correlationcomparison}(left). The same fact holds for $(2/3,1/3,2/3)$ or  $(1/3,2/3,2/3)$ when the Werner state pair in the isotropic joining scenario describes $A$-$C$ or $A$-$B$, respectively; 
in both cases, we would have obtained yet another cone in Fig. \ref{fig:correlationcomparison} that sits on a face of the classical tetrahedron boundary.

Having determined the joinable trios of both Werner and isotropic states, we are now in a position to also answer the question of what {\em pairs} $A$-$B$ and $A$-$C$ of states are joinable with one another. In the Werner state case, this is obtained by projecting the Werner joinability bicone down to the $\Psi^{-}_{AB}$-$\Psi^{-}_{AC}$ plane, resulting in the following:

\begin{coro}
Two pairs of qudit Werner states with parameters $\Psi^{-}_{AB}$ and
$\Psi^{-}_{AC}$ are joinable if and only if
$\Psi^{-}_{AB},\Psi^{-}_{AC}\geq-\frac{1}{2}$, or if the parameters satisfy
\begin{equation}
\label{eq:ellipse}
 (\Psi^{-}_{AB}+\Psi^{-}_{AC})^2+\frac{1}{3}(\Psi^{-}_{AB}-\Psi^{-}_{AC})^2 \leq 1,
\end{equation}
or additionally, in the case $d\geq3$, if
$\Psi^{-}_{AB},\Psi^{-}_{AC}\leq\frac{1}{2}$.
\end{coro}

For isotropic states, we may similarly project the cone of Eq. (\ref{d1isoa}) onto the 
$\Phi^{+}_{AB}$-$\Phi^{+}_{AC}$ plane to obtain the 1-2 joining boundary. This yields the following:

\begin{coro}
\label{isoCoro}
Two pairs of qudit isotropic states with parameters $\Phi^{+}_{AB}$ and
$\Phi^{+}_{AC}$ are joinable if and only if they lie within the convex hull of the ellipse
\begin{equation}
\label{eq:ellipseiso}
 \frac{(\Phi^{+}_{AB}/d+\Phi^{+}_{AC}/d-1)^2}{(1/d^2)} +
 \frac{(\Phi^{+}_{AB}/d-\Phi^{+}_{AC}/d)^2}{ (d^2-1)/d^2 } = 1,
\end{equation}
and the point $(\Phi^{+}_{AB},\Phi^{+}_{AC})=(0,0)$.
\end{coro}

\begin{figure}[t]
\subfigure[]{
\includegraphics[width=.46\columnwidth,viewport=0 0 2300 2200,clip]{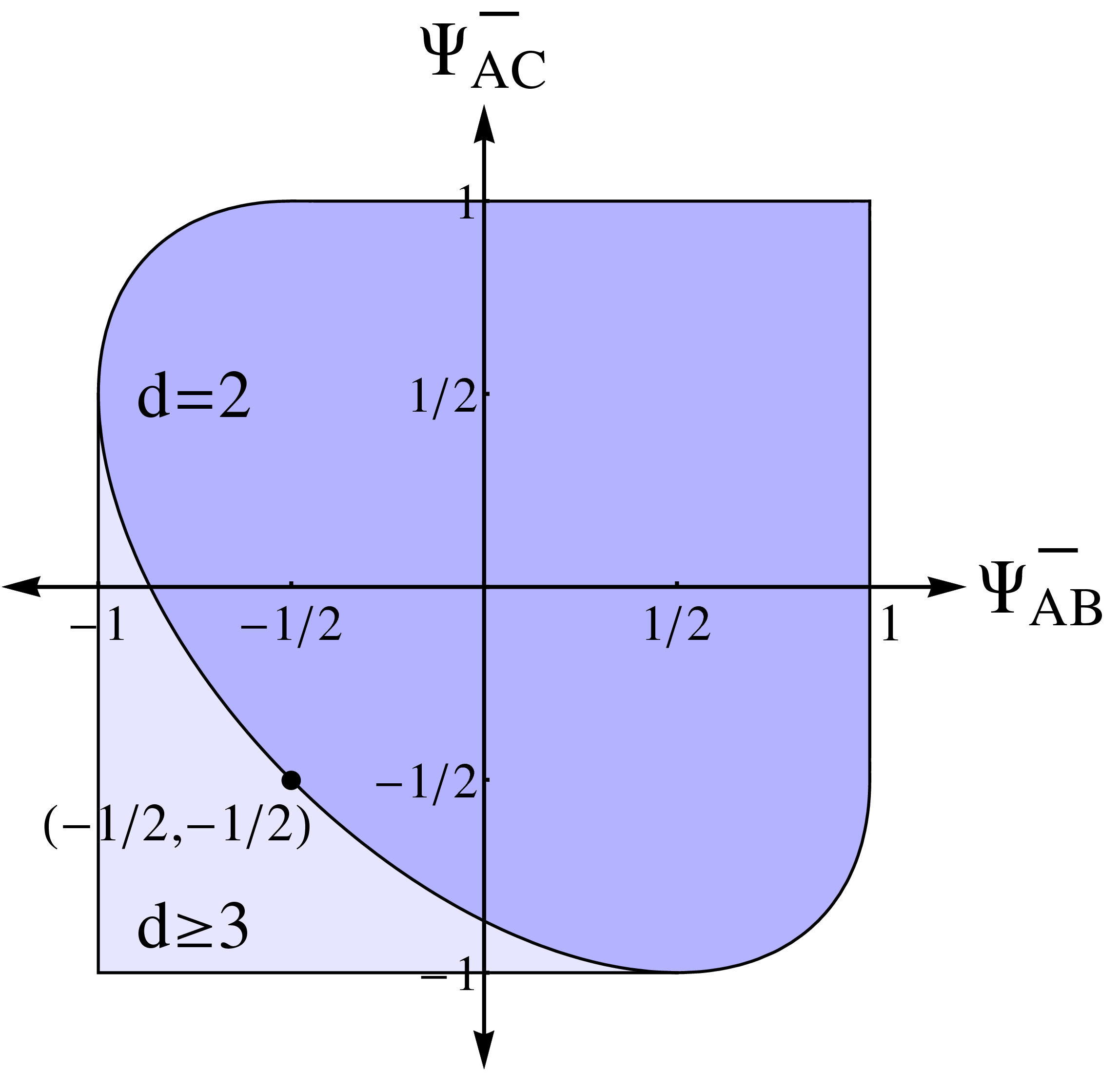}
\label{fig:wernerpairjoinability}}
\subfigure[]{
\includegraphics[width=.46\columnwidth,viewport=0 0 1200 1200,clip]{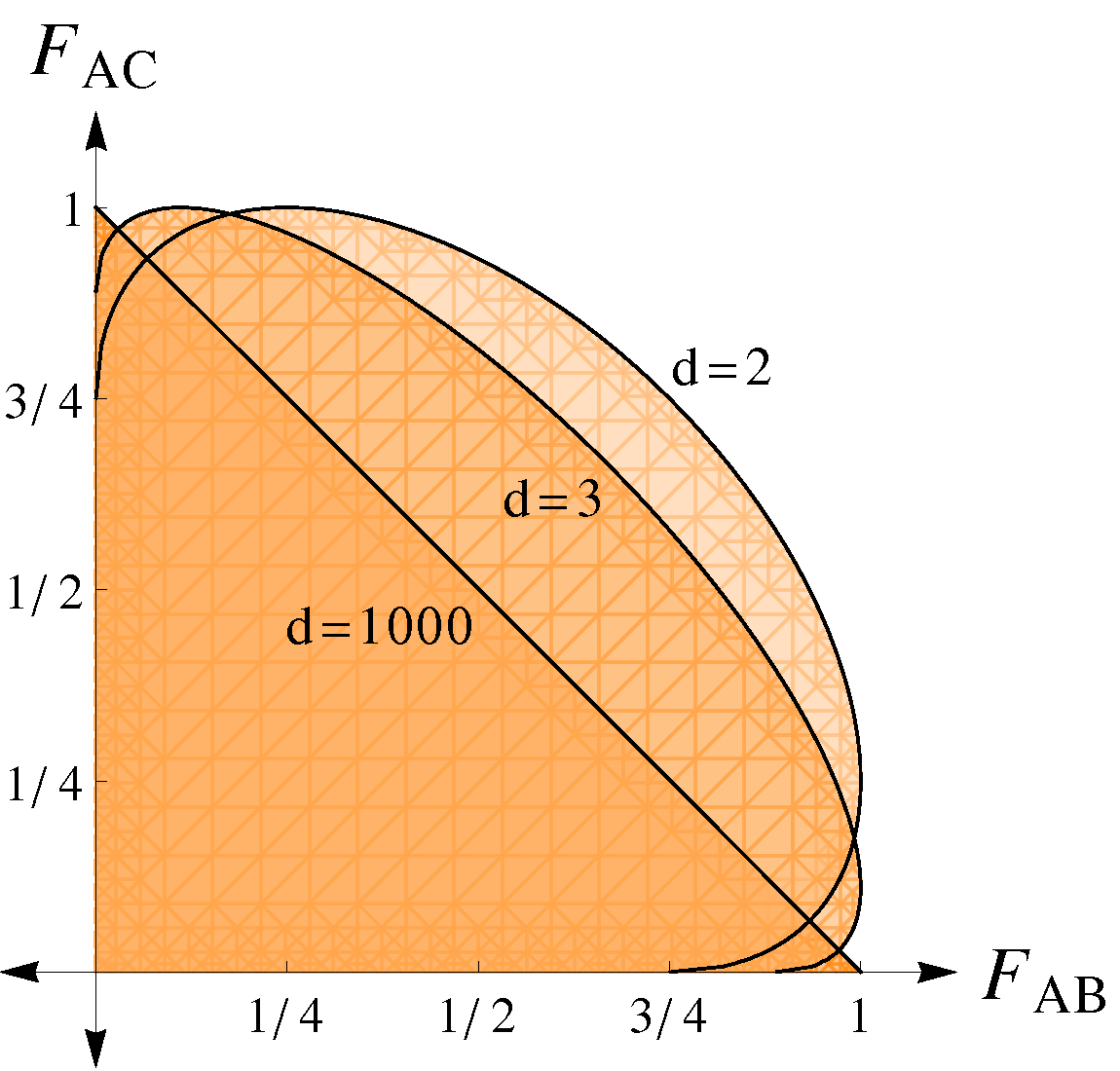}
\label{fig:isotropicpairjoinability}}
\vspace*{-1mm}
\caption{(Color online) 
Two-party joinability limitations for Werner and isotropic qudit states. 
(a) Werner states. The shaded region corresponds to joinable Werner pairs, with the lighter 
region being valid only for $d\geq3$. The rounded boundary is the ellipse determined by Eq. (\ref{eq:ellipse}). This explicitly shows the existence of pairs of entangled Werner states that are within the circular boundary determined by the weak CKW inequality, Eq. (\ref{CKWweak}), 
yet are {\em not} joinable. 
(b) Isotropic states. The three regions correspond to the joinable pairs of isotropic states for $d=2$, $d=3$ and $d=1000$. This shows how, in the limit of large $d$, the trade-off in isotropic state parameters becomes linear, consistent with known results on $d$-dimensional quantum cloning 
\cite{Cerf2000}. }
\end{figure}

Lastly, by a similar projection of the isotropic cone given by Eqs. (\ref{d1iso})-(\ref{d1isoa}), we may explicitly characterize 
the Werner-isotropic hybrid 1-2 joining boundary:

\begin{coro}
\label{projectioncorohybrid}
An isotropic state with parameter $\Phi^{+}_{AB}$ and a Werner state with parameter 
$\Psi^{-}_{BC}$ are joinable if and only if they lie within the convex hull of the ellipse
\begin{equation}
\label{eq:ellipsehybrid}
 \frac{(\Phi^{+}_{AB}/d+\Phi^{+}_{AC}/d-1)^2}{(1/d^2)}+
 \frac{(\Phi^{+}_{AB}/d-\Phi^{+}_{AC}/d)^2}{(d^2-1)/d^2} = 1,
\end{equation}
and the point $(\Phi^{+}_{AB},\Psi^{-}_{BC})=(0,1)$, and, for $d\geq 3$, within 
the additional convex hull introduced by the point $(\Phi^{+}_{AB},\Psi^{-}_{BC})=(0,1)$.
\end{coro}

The above results give the exact quantum-mechanical rules for the
two-pair joinability of Werner and isotropic states, as pictorially summarized in
see Figs. \ref{fig:wernerpairjoinability} and \ref{fig:isotropicpairjoinability}. 
A number of interesting features are worth noticing. First, by restricting to the 
line where $\Psi^{-}_{AB}=\Psi^{-}_{AC}$, we can conclude that qubit Werner states 
are $1$-$2$ sharable if and only if $\Psi^{-}\geq -1/2$, whereas for
$d\geq 3$, all qudit Werner states are $1$-$2$ sharable. As we shall
see, this agrees with the more general analysis of Sec. \ref{subsec:Wsharability}.

Second, some insight into the role of entanglement in limiting
joinability may be gained. In the first quadrant of Fig.
\ref{fig:wernerpairjoinability}, where neither pair is entangled, it
is no surprise that no joinability restrictions apply. Likewise, it
is not surprising to see that, in the third quadrant where both pairs
are entangled, there is a trade-off between the amount of entanglement
allowed between one pair and that of the other. But, in the second and
fourth quadrants we observe a more interesting behavior. Namely, these
quadrants show that there is also a {\em trade-off between the amount
of classical correlation in one pair and the amount of entanglement in
the other pair}. In fact, the smoothness of the boundary curve as it
crosses from one of the pairs being entangled to unentangled suggests
that, at least in this case, entanglement is not the correct figure of
merit in diagnosing joinability limitations.

\subsection{Isotropic joinability results from quantum cloning}
\label{subsec:cloning}

Interestingly, the above results for 1-2 joinability of isotropic states can also be 
obtained by drawing upon existing results for asymmetric quantum cloning, see e.g. \cite{Cerf2000,Iblisdir2005} for 1-2 and 1-3 asymmetric cloning and \cite{Ramanathan2009, Jiang2010,Kay2012} for 1-$n$ asymmetric cloning. One approach to obtaining the optimal asymmetric cloning machine is to exploit the Choi isomorphism \cite{Choi1975} to translate the construction of the optimal cloning map to the construction of an optimal operator (or a ``telemapping state''). This connection is made fairly clear in \cite{Ramanathan2009,Kay2012}; in particular, ``singlet monogamy'' refers to the trade-off in fidelities of the optimal 1-$n$ asymmetric cloning machine or, equivalently, to the trade-off in singlet fractions for a $(1+n)$ qudit state. We describe 
how the approach to solving the optimal 1-$n$ asymmetric cloning problem may be rephrased to solve the 1-$n$ joinability problem for isotropic states.

The state $\ket{\Psi}$ described in Eq. (4) of \cite{Ramanathan2009} is a 1-$n$ joining state for $n$ isotropic states characterized by singlet fractions $F_{0,j}$ (related to the isotropic state parameter by $F_{0,j}=\Phi^{+}_{0,j}/d$, as noted). The bounds on the singlet fractions are determined by the normalization condition of $\ket{\Psi}$, together with the requirement that 
$\ket{\Psi}$ be an eigenstate of a certain operator $R$ defined in Eq. (3) of \cite{Ramanathan2009}. That $\ket{\Psi}$ is an isotropic joining state is readily seen from its construction, and that it may optimize the singlet fractions (hence delineate the boundary in the $\{F_{0,j}\}$ space) is proven in \cite{Kay2012}.
Our contribution here is the observation that this result provides the solution to the 1-$n$ joinability of isotropic states. The equivalence is established by the fact that optimality is preserved in either direction by the Choi isomorphism. 

Quantitatively, the boundary for 1-$n$ optimal asymmetric cloning, is given by Eq. (6) in \cite{Ramanathan2009} in terms of singlet fractions. Specializing to the 1-2 joining case and rewriting in terms of $\Phi^{+}$, we have 
$$ \Phi^{+}_{AB}+\Phi^{+}_{AC}\leq (d-1) + 
\frac{1}{n+d-1}\Big(\sqrt{\Phi^+_{AB}}+\sqrt{\Phi^{+}_{AC}}\Big)^2. $$
\noindent
As one may verify, this is equivalent to the result of Corollary \ref{isoCoro}.
In light of this connection, the fact that, as $d$ increases, the isotropic-joinability cone of Fig. \ref{fig:correlationcomparison}(right) 
becomes flattened down to the $\alpha_{AB}$-$\alpha_{AC}$ plane is directly related to 
the \emph{linear trade-off} in the isotropic state parameters for the semi-classical limit 
$d\rightarrow\infty$, as discussed in \cite{Cerf2000}. Within our three-party joining picture, 
we can give a partial explanation of this fact: namely, it is a consequence of the classical joining boundary in tandem with the upper limit on the agreement parameter $\alpha_{BC}$ for the Werner state on $B$-$C$: $\alpha_{BC}\leq 2/(d+1)$. In the limit of $d\rightarrow\infty$, these two boundaries conspire to limit the ($A$-$B$)-($A$-$C$) isotropic state joining boundary to a triangle, as explicitly seen in Fig. \ref{fig:isotropicpairjoinability}(right).

For the general 1-$n$ isotropic joining scenario, the quantum-cloning results additionally 
imply the following:
\begin{thm} 
A list of $n$ isotropic states characterized by parameters $\Phi_{0,1}^+,\ldots,\Phi_{0,j}^+$ 
is 1-$n$ joinable if and only if the (positive-valued) parameters satisfy
\begin{equation}
\label{eq:1njoiningiso}
\sum_{j=1}^n \Phi^{+}_{0,j}\leq (d-1) + \frac{1}{n+d-1}\bigg(\sum_{j=1}^n \sqrt{\Phi^+_{0,j}}\bigg)^2.
\end{equation}
\label{thm:1njoiningiso}
\end{thm}
\noindent 
Interestingly, similar to our discussion surrounding Eq. (\ref{CKWweak}), the authors of \cite{Ramanathan2009} argue how the ``singlet monogamy'' bound can lead to stricter predictions (e.g., on ground-state energies in many-body spin systems) than the standard monogamy of entanglement bounds based on CKW inequalities \cite{Wootters2000,Verstraete2006}.

\subsection{Sharability of Werner and isotropic qudit states}
\label{subsec:Wsharability}

We next turn to sharability of Werner and isotropic states in $d$ dimension, beginning from  
the important case of $1$-$n$ sharing. For Werner states, a proof based on a representation-theoretic approach is given in Appendix \ref{sec:sharingreptools}. Although we expect a similar proof to exist for isotropic states, we obtain the desired $1$-$n$ sharability result by building on the relationship with quantum cloning problems highlighted above. We then outline a constructive procedure for determining the more general $m$-$n$ sharability of Werner states.

Our main results are contained in the following:
\begin{thm}
\label{thm:1nsharability}
A qudit Werner state with parameter $\Psi^{-}$ is $1$-$n$ sharable if and
only if
\begin{equation}
\Psi^{-} \geq -\frac{d-1}{n}.
\label{Wsharing} 
\end{equation}
\end{thm}

\begin{thm}
\label{thm:1nsharabilityiso}
A qudit isotropic state with parameter $\Phi^{+}$ is $1$-$n$ sharable if and
only if
\begin{equation}
\Phi^{+} \leq 1+\frac{d-1}{n}.
\label{isosharing} 
\end{equation}
\end{thm}
\begin{proof}
Specializing Eq. (\ref{eq:1njoiningiso}) to the case of equal parameters for all $n$ isotropic states, the above result immediately follows. As stated in \cite{Ramanathan2009}, this is consistent with the well known result for optimal 1-$n$ symmetric cloning.
\end{proof}
A pictorial representation of the above sharability results, specialized to qubits, is given in Fig. \ref{fig:sharability}. In the case of Werner state sharing, Eq. (\ref{Wsharing}) implies that a \emph{finite} parameter range exists where the corresponding Werner states are \emph{not} sharable. In contrast, for $d\geq 3$, every Werner state is at least $1$-$2$ sharable. This simply reflects the fact that $\ket{\psi^-_d}$ (recall Eq. (\ref{eq:purewerner})) provides a 1-$(d-1)$ sharing state for a most-entangled qudit Werner state. With isotropic state sharing, for all $d$ there is, again, a {\em finite} range of isotropic states which are not sharable.

\begin{figure}[b]
\begin{subfigure}
{\includegraphics[width=.9\columnwidth,viewport=10 550 570 700,clip]{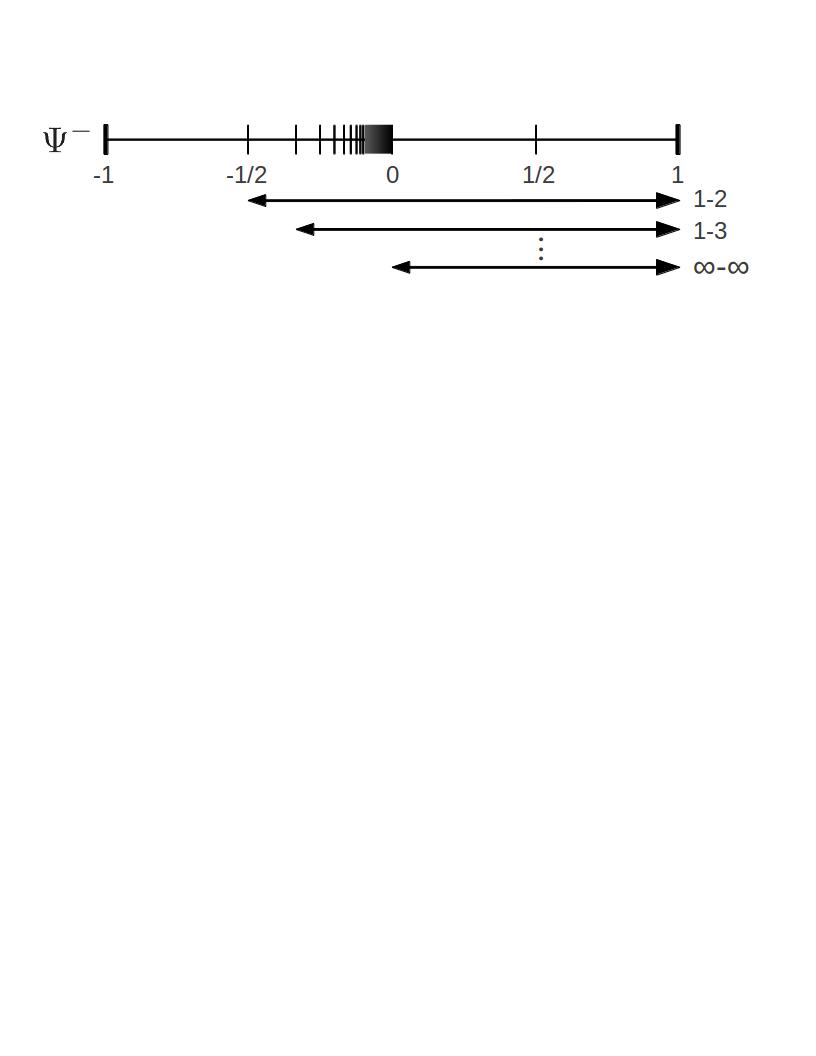}}
\end{subfigure}
\begin{subfigure}
{\includegraphics[width=.9\columnwidth,viewport=20 550 570 700,clip]{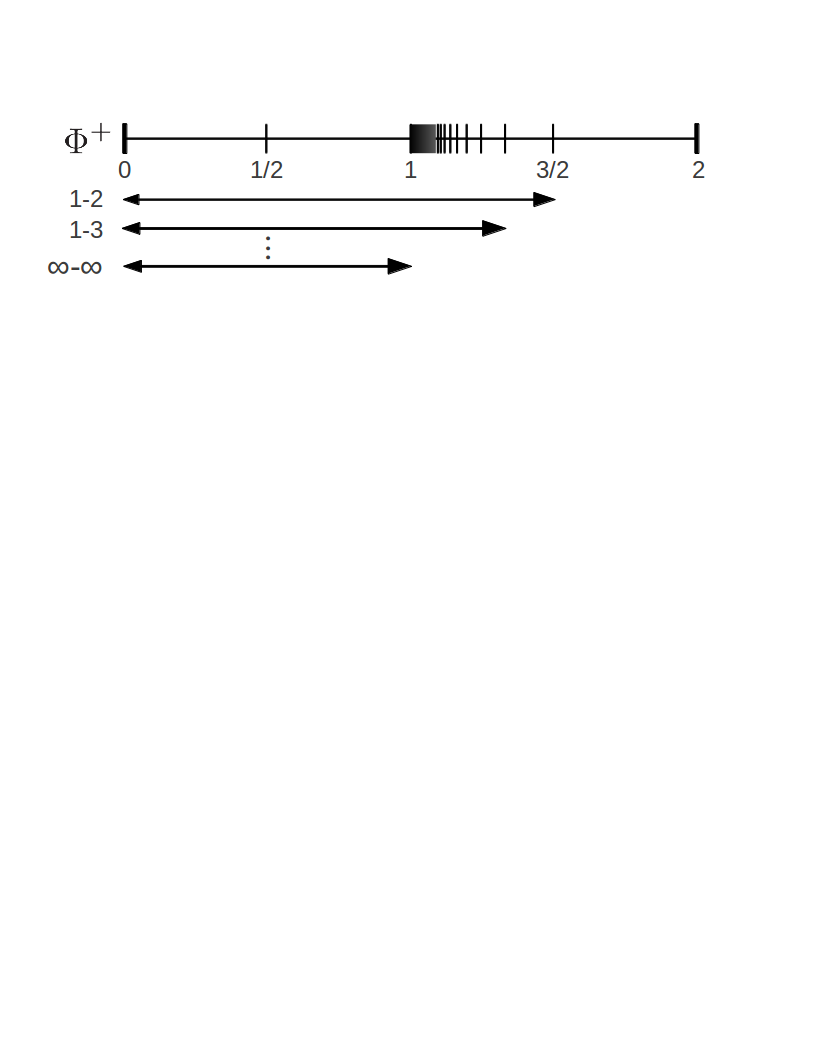}}\end{subfigure}
\vspace*{-5mm}
\caption{Pictorial summary of sharability properties of qubit Werner and isotropic
states, according to Eqs. (\ref{Wsharing}) and (\ref{isosharing}). The arrow-headed lines
depict the parameter range for which states satisfy each of the
sharability properties displayed to the right and left, respectively. The vertical ticks
between end points of these ranges indicate the points at which subsequent $1$-$n$ sharability
properties begin to be satisfied. 
}
\label{fig:sharability}
\end{figure}

The simplicity of the results in Eqs. (\ref{Wsharing})-(\ref{isosharing}) is intriguing and begs for intuitive interpretations. Consider a central qudit surrounded by $n$ outer qudits. If the central qudit is in the same Werner or isotropic state with each outer qudit, then Theorems \ref{thm:1nsharability} and \ref{thm:1nsharabilityiso} can be reinterpreted as providing a bound on the sums of concurrences. For Werner states, we have that the sum of all the central-to-outer concurrences cannot exceed the number of modes by which the systems may disagree (i.e., $d-1$). In the isotropic state case, the sum of the $n$ pairwise concurrences cannot exceed the maximal concurrence value given by $\mathcal{C}_{\text{max},d}=\sqrt{{2(d-1)}/d}$. These rules do not 
hold in more general joining scenarios, as we already know from Sec. \ref{subsec:Wjoin}. There, we found that the trade-off between $A$-$B$ concurrence and $A$-$C$ concurrence is not a linear one, as such a simple ``sum rule'' would predict; instead, it traces 
out an ellipse (recall Fig. \ref{fig:wernerpairjoinability}).

Starting from the proof of Thm. \ref{thm:1nsharability} found in Appendix \ref{sec:sharingreptools}, in conjunction with similar representation-theoretic tools, it is possible to devise a constructive algorithm for determining the $m$-$n$ sharability of Werner states. The basic observation is to realize that the most-entangled $m$-$n$ sharable Werner state corresponds to the largest eigenvalue of a certain Hamiltonian operator $H_{m,n}$, which is in turn expressible in terms of Casimir operators. Calculation of these eigenvalues may be obtained using Young diagrams. Although we lack a general closed-form expression for $\text{max}(H_{m,n})$, the required calculation can nevertheless be performed numerically. Representative results for $n$-$m$ sharability of low-dimensional Werner states are shown in Table \ref{tab:nmsharability}.


\begin{widetext}

\begin{table}[t]
\caption{Exact results for $n$-$m$ sharability of Werner states for
different subsystem dimension, with $m$ and $n$ increasing from left
to right and from top to bottom in each table, respectively. For each
sharability setting, the value $-\Phi$ is given. Asterisks correponds 
to entries whose values have not been explicitly computed. 
 }
\vspace*{1mm}
\hspace*{-0.0cm}(a) $d=2$  \hspace*{3.5cm} (b) $d=3$ \hspace*{3cm} (c) $d=4$ \\
\vspace*{1mm}
\begin{tabular}{c|c|c|c|c|c}
$n,m$ & 1     & 2    & 3    & 4    & 5   \\ \hline 
  1   & 1     & 1/2 & 1/3 & 1/4 & 1/5  \\ \hline
  2   & 1/2  & 1/2 & 1/3 & 1/4 & 1/5  \\ \hline
  3   & 1/3  & 1/3 & 1/3 & 1/4 & 1/5  \\ \hline
  4   & 1/4  & 1/4 & 1/4 & 1/4 & 1/5  \\ \hline
  5   & 1/5  & 1/5 & 1/5 & 1/5 & 1/5  
\end{tabular}\hspace*{1cm} \begin{tabular}{c|c|c|c|c|c}
$n,m$ & 1     & 2    & 3    & 4    & 5   \\ \hline 
  1   & 1     & 1    & 2/3 & 1/2 & 2/5  \\ \hline
  2   & 1     & 1/2 & 1/2 & 1/2 & 2/5  \\ \hline
  3   & 2/3  & 1/2 & 1/3 & 1/3 & 1/3  \\ \hline
  4   & 1/2  & 1/2 & 1/3 & 1/4 & 1/4  \\ \hline
  5   & 2/5  & 1/3 & 1/3 & 1/4 &  $\ast$ 
\end{tabular}\hspace*{1cm}\begin{tabular}{c|c|c|c|c|c}
$n,m$ & 1     & 2    & 3    & 4    & 5   \\ \hline 
  1   & 1     & 1    & 1    & 3/4 & 3/5  \\ \hline
  2   & 1     & 1    & 2/3 & 1/2 & 1/2  \\ \hline
  3   & 1     & 2/3 & 5/9 & 1/2 & $\ast$        \\ \hline
  4   & 3/4  & 1/2 & 1/2 &  $\ast$      &  $\ast$       \\ \hline
  5   & 3/5  & 1/2 &   $\ast$     &   $\ast$     &  $\ast$  
\end{tabular}
\label{tab:nmsharability}
\end{table}

\end{widetext}

\section{Further remarks}
\label{sec:furtherremarks}

\subsection{Joinability beyond the three-party scenario}

In Sec. \ref{sec:Wernerstates}, we focused on considering joinability of 
three bipartite (Werner or isotropic) states in a ``triangular fashion'', namely, 
relatively to the simplest overlapping neighborhood structure ${\cal N}_1= \{A, B\}$, 
${\cal N}_2= \{A, C\}$ on ${\cal H}^{(3)}$. In a more general $N$-partite scenario,
other neighborhood structures and associated joinability problems may naturally 
emerge. For instance, we may want to answer the following question: 
Which sets of $N(N-1)/2$ Werner-state (or isotropic-state) pairs are joinable? 
The approach to solving this more general problem parallels
the specific three-party case we discussed.

If each pair is in a Werner state, then if a joining state exists, there 
must exist a joining state with collective invariant symmetry (that is, 
invariant under arbitrary collective unitaries $U^{\otimes N}$).
Thus, we need only look in the set of states respecting this
symmetry. Any such operator may be
decomposed into a sum of operators, which each have support on just a
single irreducible subspace. This is useful because positivity of the
joining operator when restricted to each irreducible subspace is
sufficient for positivity of the overall operator. The joining
operators may then be decomposed into the projectors on each irreducible
subspace and corresponding bases of traceless operators on the
projectors. The basis elements will be combinations of permutation
operators and the dimension of each such operator subspaces is
given by the square of the hook length of the corresponding Young
diagram \cite{Fuchs1997}. The remaining task is to obtain a characterization 
of the positivity of the operators on {\em each} irreducible subspace. In 
\cite{Byrd2003}, for example, a method for characterizing the positivity of 
low-dimensional operator spaces is presented. As long as the number of subsystems
remains small, this approach grants us a 
computationally friendly characterization of positivity of the joining states. 
The bounds on the joinable Werner pairs may then be obtained by projecting the positivity
characterization boundary onto the space of Werner pairs, analogous to
the space of Fig. \ref{fig:wernerpairjoinability}.

While a complete analysis is beyond our scope, a similar method may in principle be followed 
to determine more general joinability bounds for isotropic states. 
However, a twirling operation that preserves the joining property only exists for certain isotropic joining 
scenarios. For instance, we took this issue into consideration when we required the $B$-$C$ system to 
be in a Werner state while $A$-$B$ and $A$-$C$ were isotropic states; it would not have been possible to 
take the same approach if all three pairs were isotropic states.

\subsection{Sharability of general bipartite qubit states}

For qubit Werner states, one can use the methods of the proof of Thm. III.6 to show that $1$-$n$ sharability does imply $n$-$n$ sharability [cf. Table I.(a)]. This property neither holds for Werner qudit states nor bipartite qubit states in general. The simplest example of a Werner state which disobeys this property is the most-entangled qutrit Werner state 
$\rho(\Psi^{-}=-1)_{d=3}$. This state is 1-2 sharable, as evidenced by the point $(-1,-1,-1)$ lying within the bi-cone described by Eq. (\ref{d1}). The corresponding sharing state is the totally antisymmetric state on three qutrits as given by Eq. (\ref{eq:purewerner}). This is the unique sharing state because the collective disagreement between the subsystems of each joined bipartite Werner state forces collective disagreement among the subsystems of the tripartite joining state; the totally antisymmetric state is the only quantum state satisfying this property. Since the only 1-2 sharing state for $\rho(\Psi^{-}=-1)_{d=3}$ is pure and entangled, clearly there can exist no 2-2 sharing.

Additionally, we present below a counter-example that involves qubit states off the Werner line:

\begin{prop}
For a generic bipartite qubit state $\rho$, $1$-$n$ sharability does
not imply $n$-$n$ sharability.
\end{prop}

\begin{proof}  
We claim that the following bipartite state on two qubits, 
$$\rho=\frac{1}{3}\big[ \big(\ket{00}+\ket{11}\big)
\big(\bra{00}+\bra{11}\big)+\ketbra{10}\Big] \equiv \rho_{L_1 R_1},$$ 
\noindent 
is 1-2 sharable but {not} 2-2 sharable. To show that $\rho$ is $1$-$2$
sharable, direct calculation shows that the two relevant partial-trace constraints uniquely 
identify $w_3 \equiv \ketbra{\psi}$ as the only valid sharing state, with
$$\ket{\psi}\equiv \frac{1}{\sqrt{3}}(\ket{000}+\ket{101}+
\ket{110}).$$
\noindent 
The above state may in turn be equivalently written as 
$$\ket{\psi}=\frac{1}{\sqrt{3}}\ket{0}\otimes\ket{00}+
\sqrt{\frac{2}{3}}\ket{1}\otimes\frac{1}{\sqrt{2}}
\left(\ket{01}+\ket{10}\right).$$ 
\noindent 
In order for $\rho$ to be 2-2 sharable, a four-partite state $w_4$
must exist, such that Tr$_{\hat{L_i}\hat{L_j}}(w_4)= \rho$, for
$i,j=1,2$. Any state which 2-2 shares $\rho$ must then 1-2 share the
pure entangled state $w_3$. That is, in constructing the 2-2 sharing
state for $\rho$, we bring in a fourth system $L_2$ which must reduce
(by tracing over $L_1$ or $L_2$) to $w_3$. But, since $w_3$ is a pure
entangled state, it is not sharable. Thus, there cannot exist a 2-2
sharing state for $\rho$.
\end{proof}

We conclude by stressing that our Werner and isotropic state sharability results allow in fact to put bounds 
(though not necessarily tight ones) on the sharability of an \emph{arbitrary bipartite qudit state}. 
It suffices to observe that any bipartite state can be transformed into a Werner or isotropic state by the 
action of the respective twirling map (either Eq. (\ref{twirlapp}) or (\ref{twirliso})). Theorem \ref{thm:LOCC} 
proves that the sharability of a state cannot be decreased by a unitary mixture map, 
and hence twirling cannot decrease sharability. This thus establishes the following:

\begin{coro}
A bipartite qudit state $\rho$ is no more sharable
than the Werner state $$\tilde{\rho}\equiv \int U\otimes U\rho_{V}
U^{\dagger}\otimes U^{\dagger} d \mu(U),$$ and the isotropic state 
$$\bar{\rho}\equiv \int U^*\otimes U\rho_{V}
U^{T}\otimes U^{\dagger} d \mu(U),$$
for any $\rho_{V}=\identity\otimes V\rho\,\identity\otimes V^{\dagger}$, 
with $V\in \mathfrak{U}(d)$. 
\end{coro}

In the qubit case, for instance, any maximally entangled pure state
can be transformed into $\ket{\Psi^{-}}$ or $\ket{\Phi^{+}}$ by the action of some local
unitary $\identity\otimes V$. Thus, all maximally entangled pure qubit
states and their ``pseudo-pure'' versions, obtained as mixtures with
the fully mixed states, have the same sharability properties as the
Werner/isotropic states.

\section{Conclusions}
\label{sec:end}

We have provided a general framework for defining the notions of quantum joinability and sharability in finite-dimensional 
multipartite quantum systems, and compared both to the analogous classical notions. Special emphasis has been given to identifying the role of entanglement in both scenarios. In order to obtain mathematically necessary and sufficient conditions, we have specifically analyzed the three-party joinability and the $m$-$n$ sharability properties of qudit Werner and isotropic states. We found that the entanglement content of the joined bipartite states does \emph{not} suffice to determine the resulting joinability properties. Additionally, we analyzed the role that the classical joining limitations play in restricting quantum joining in these scenarios. As a byproduct, this led to an explanation for the linear trade-off in singlet fractions (or cloning fidelities) as $d\rightarrow \infty$. We further determined simple analytical expressions for 1-$n$ sharability, namely, that the sum of the bipartite concurrences cannot exceed $(d-1)$ for 
Werner states and cannot exceed $\mathcal{C}_{\text{max},d}$ for isotropic states. In the more general case of $m$-$n$ sharing, we laid out an algorithmic procedure for calculating the most-entangled $m$-$n$ sharable Werner states using Young diagrams. As a corollary of our results, we established upper bounds on the sharability of \emph{arbitrary} bipartite qudit states.

Several open questions remain for future investigations. The use of the Choi isomorphism in translating between 1-$n$ joining scenarios and cloning scenarios points to an intriguing connection between the mathematical structures of reduced quantum states and reduced channels. Pushing this connection further, we believe it would be fruitful to investigate the analogous problem of {\em joining quantum channels}, making contact, in particular, with recent work on extending quantum operations \cite{ReebJMP}.

Also related to channels, we would like to continue investigating the effects of different maps on the sharability properties of the input states. In Thm. \ref{thm:LOCC}, we showed that mixtures of local unitaries cannot decrease sharability. It would be interesting to obtain a characterization of the set of maps which do not decrease sharability and compare them to LOCC maps which are known to not increase entanglement. Since we do not have a counterexample to the statement, we conjecture that sharability cannot be decreased by {\em any} LOCC map.

Keeping with the approach pursued here, further progress toward obtaining necessary and sufficient joinability conditions may be made by narrowing the set of states to be analyzed to other physically relevant families and/or specific joining structures. For instance, families of mixed qudit states may arise as reduced states of many-body ground states of spin-1/2 or higher spin Hamiltonians parametrized by an external control parameter. In this context, it may be insighful to examine what joinability and sharability features of quantum correlations change as the system is driven across a quantum phase transition, complementing extensive investigation of ground-state entanglement \cite{Fazio} and generalized entanglement \cite{Barnum2004,Somma2004} in critical phenomena. Finally, since generalized entanglement is defined {\em without} relying on a preferred tensor-product decomposition, with ``generalized reduced states'' being constructed through a suitable reduction map relative to 
observable 
subspaces \cite{Barnum2004,Barnum2005}, a natural question arises: What is the nature and role of joinability limitations \emph{beyond subsystems}? We leave exploration of this intriguing question to future research.

\begin{acknowledgments}
We are especially indebted to Benjamin Schumacher for inspiring discussions and early input to this work, and to Marco Piani for constructive feedback and for bringing Ref. \cite{Eggeling2001} to our attention. We thank Michael Walter for sharing his notes and work on the 1-$n$ sharability problem and Lin Chen for pointing out an inaccuracy in the original version of the preprint. It is a pleasure to also thank Gerardo Adesso, Sergio Boixo, Guilio Chiribella, Chris Fuchs, Rosa Orellana, Francesco Ticozzi, Reinhard Werner and Bill Wootters for valuable discussions and input throughout this work. Partial support from the Constance and Walter Burke Special Projects Fund in Quantum Information Science is gratefully acknowledged. 
\end{acknowledgments}

\begin{appendix}

\section{Proofs of three-party quantum joinability problems}
\label{sec:joiningreptools}

As remarked in Sec. II, joinability limitations are a manifestation of the non-negativity constraint placed on the joining density operators. 
Symmetries of the reduced states allow us to both narrow our search for a joining state and to obtain simple characterizations of non-negativity. For both Werner and isotropic states, the relevant search space for valid three-partite joining states may be characterized in terms of suitable subsystem permutation operators, according to Eqs. (\ref{permrepspan}) and (\ref{eq:ptrep}), respectively. Such parameterizations do not offer, however, a straightforward characterization of non-negativity. To this end, we need to choose a different basis for which non-negativity is more simply expressed in terms of its (the basis') coefficients. The key idea is to decompose the operator space into subspaces for which the non-negativity of a given operator's projection into each subspace ensures the non-negativity of the given operator. This is achieved if the operators within each operator subspace act non-trivially on orthogonal vector subspaces. Irreducible representations (irreps) provide such a decomposition. 

For the symmetric group $S_N$, the irrep subspaces are projected into by the so-called Young symmetrizers. The prescription for constructing Young symmetrizers from the permutation representations may be found in most books on representation theory, see e.g. \cite{Fulton1991,Weyl1997}. In our case, $N=3$, the group $S_3$ has three inequivalent 
irreps, and the relevant Young symmetrizers read
\begin{eqnarray}
R_+&\hspace*{-1mm}=\hspace*{-1mm}&
\frac{1}{6}\Big[\mathbb{I}+V_{(AB)}+V_{(BC)}+V_{(CA)}+V_{(ABC)}+V_{(CBA)}\Big] ,
\nonumber \\
R_-&\hspace*{-1mm}=\hspace*{-1mm}&
\frac{1}{6}\Big[\mathbb{I}-V_{(AB)}-V_{(BC)}-V_{(CA)}+V_{(ABC)}+V_{(CBA)}\Big],
\nonumber \\
R_0&\hspace*{-1mm}=\hspace*{-1mm}&
\frac{1}{3}\Big[ 2\mathbb{I}-V_{(ABC)}-V_{(CBA)}\Big],
\label{Young}
\end{eqnarray}
where, by following Ref. \cite{Eggeling2001}, we use cycle notation to label 
permutation elements, and an orthonormal basis of Pauli-like operators acting on the 
support of $R_0$ is
\begin{eqnarray}
R_1&=&\frac{1}{3}\Big[2V_{(BC)}-V_{(CA)}-V_{(AB)}\Big],\nonumber \\
R_2&=&\frac{1}{\sqrt{3}}\Big[ V_{(AB)}-V_{(CA)}\Big],\nonumber \\
R_3&=&\frac{i}{\sqrt{3}}\Big[ V_{(ABC)}-V_{(CBA)}\Big]. 
\label{Rs}
\end{eqnarray}
Defining $r_k(w)=\tr{}{w R_k}$, any $w \in \mathcal{W}_W$ is of the form
\begin{align}
\label{generalstate}
w(\vec{r})=&\frac{6r_+}{d(d+1)(d+2)}R_++\frac{6r_-}{d(d-1)(d-2)}R_-
\nonumber \\ & + \frac{3}{2d(d^2-1)}\sum_{i=0}^3 r_i R_i,
\end{align}
where $\vec{r}=(r_+,r_-,\ldots,r_3)$. Normalization is ensured by
\begin{equation*}
\tr{}{w\identity}=\tr{}{w(R_++R_-+R_0)}=r_++r_-+r_0=1, 
\end{equation*}
while non-negativity is given by the following simple relationships:
\begin{equation}
\label{eq:nonnegconstraints}
r_+,r_-,r_0\geq 0, \;\;\; r_1^2+r_2^2+r_3^2\leq r_0^2.
\end{equation}

Following again Ref. \cite{Eggeling2001} (see in particular Sec. VI.A), the analogous decomposition of operators of the form given in Eq. (\ref{eq:ptrep}) into orthogonal 
projectors reads 
\begin{align}
S_+&\hspace*{-1mm}=\frac{1}{2}\Big[\identity+V_{(BC)}-\frac{V_{(AB)}^{T_A}+V_{(AC)}^{T_A}+V_{(ABC)}^{T_A}+V_{(CBA)}^{T_A}}{d+1}\Big],
\nonumber \\
S_-&\hspace*{-1mm}=\frac{1}{2}\Big[\identity-V_{(BC)}+\frac{V_{(ABC)}^{T_A}+V_{(CBA)}^{T_A}-V_{(AB)}^{T_A}-V_{(AC)}^{T_A}}{d-1}\Big] ,
\nonumber \\
S_0&\hspace*{-1mm}=\frac{1}{d^2-1}\Big[d(V_{(AB)}^{T_A}+V_{(AC)}^{T_A})-(V_{(ABC)}^{T_A}+V_{(CBA)}^{T_A})\Big],
\label{Youngiso}
\end{align}
whereas an orthonormal basis of operators acting on the support of $S_0$ is
\begin{align}
S_1=&\frac{1}{d^2-1}\Big[d(V_{(ABC)}^{T_A}+V_{(CBA)}^{T_A})-(V_{(AB)}^{T_A}+V_{(AC)}^{T_A})\Big],\nonumber \\
S_2=&\frac{1}{\sqrt{d^2-1}}\Big[V_{(AB)}^{T_A}-V_{(AC)}^{T_A}\Big],\nonumber \\
S_3=&\frac{i}{\sqrt{d^2-1}}\Big[ V_{(ABC)}^{T_A}-V_{(CBA)}^{T_A}\Big].
\label{Rsiso}
\end{align}
In complete analogy, we define $s_k(w)=\tr{}{w S_k}$. Then any joining state 
$w \in {\mathcal W}_{\text{iso}}$ is of the form
\begin{align}
\label{generalstateiso}
w(\vec{s})=&\frac{2s_+}{d(d+2)(d-1)}S_++\frac{2s_-}{d(d-2)(d+1)}S_-
\nonumber \\ & + \frac{1}{2d}\sum_{i=0}^3 s_i S_i,
\end{align}
where $\vec{s}=(s_+,s_-,\ldots,s_3)$. Normalization is ensured by
\begin{equation*}
\tr{}{w\identity}=\tr{}{w(S_++S_-+S_0)}=s_++s_-+s_0=1, 
\end{equation*}
and non-negativity is given by
\begin{equation*}
s_+,s_-,s_0\geq 0, \;\;\; s_1^2+s_2^2+s_3^2\leq s_0^2.
\end{equation*}

Given the above decompositions of three-party joining operators, we are now equipped to formally prove our three-party joining results.

\begin{thmwj} 
\label{triowerapp} 
Three Werner qudit states with parameters $\Psi^{-}_{AB}, \Psi^{-}_{BC}, \Psi^{-}_{AC}$ are
joinable if and only if $(\Psi^{-}_{AB}, \Psi^{-}_{BC}, \Psi^{-}_{AC})$ lies within the bicone
described by
\begin{align}
1\pm\overline{\Psi^{-}}\geq \frac{2}{3}\left|\Psi^{-}_{BC}+\omega
\Psi^{-}_{AC}+\omega^2\Psi^{-}_{AB}\right| ,
\label{d1app}
\end{align}
for $d\geq3$, or within the cone described by
\begin{align}
1-\overline{\Psi^{-}}\geq \frac{2}{3}\left|\Psi^{-}_{BC}+ \omega
\Psi^{-}_{AC}+\omega^2\Psi^{-}_{AB}\right| ,
\;\;\; \overline{\Psi^{-}}\geq 0 , 
\label{d2bapp} 
\end{align}
for $d=2$, where 
\begin{equation}
\overline{\Psi^{-}}=\frac{1}{3}(\Psi^{-}_{AB}+\Psi^{-}_{BC}+\Psi^{-}_{AC}),
\;\;\;\omega=e^{i\frac{2\pi}{3}}.
\label{Phibarapp}
\end{equation}
\end{thmwj}

\begin{proof}
These joinability limitations are derived by re-expressing the non-negativity constraints of Eq. (\ref{eq:nonnegconstraints}) in terms of Werner parameters 
$\Psi^{-}_{ij}$. We have:
\begin{align*}
r_1^2+r_2^2 &= \tr{}{wR_1}^2+\tr{}{wR_2}^2 =|\tr{}{w(R_1+iR_2)}|^2\nonumber\\
&=\frac{4}{9}\left|\trb{}{w\left(V_{(BC)}+e^{i\frac{2\pi}{3}}
V_{(CA)}+e^{i\frac{4\pi}{3}}V_{(AB)}\right)}\right|^2\nonumber\\
&=\frac{4}{9}\left|\Psi^{-}_{BC}+\omega
\Psi^{-}_{AC}+\omega^2\Psi^{-}_{AB}\right|^2,
\end{align*}
and 
\begin{align}
\label{r_0}
r_0&=\tr{}{w(R_0+2R_-)}-2\tr{}{wR_-}\nonumber\\
&=\trb{}{w(\identity-\frac{1}{3}(V_{(AB)}+V_{(BC)}+V_{(CA)}))}-2r_-\nonumber\\
&=\frac{1}{3}\sum_{i<j}[(1-2r_-)-\Psi^{-}_{ij}]=1-2r_--\overline{\Psi^{-}},
\end{align}
where $\overline{\Psi^{-}}$ is defined in Eq. (\ref{Phibar}). Thus, the
spherical inequality $r_1^2+r_2^2+r_3^2\leq r_0^2$ may be rewritten as
\begin{align}
\label{sphericalineq}
(1-2r_--\overline{\Psi^{-}})^2\geq \frac{4}{9}\left|\Psi^{-}_{BC}+\omega
\Psi^{-}_{AC}+\omega^2\Psi^{-}_{AB}\right|^2+r_3^2.
\end{align}
Since a non-zero value of $r_3$ only further limits the inequality and
since the $\Psi^{-}_{ij}$ are independent of it, we maximize the range of
joinable $\Psi^{-}_{ij}$ by setting $r_3=0$.

The non-negativity is then expressed in terms of the parameters 
$\Psi^{-}_{ij}$ and $r_-$ as
\begin{align}
1-2r_--\overline{\Psi^{-}}\geq \frac{2}{3}\left|\Psi^{-}_{BC}+\omega
\Psi^{-}_{AC}+\omega^2\Psi^{-}_{AB}\right|, 
\label{gencone}\\
\overline{\Psi^{-}}+r_-\geq0 , 
\label{conebase1}\\
r_-\geq0,
\label{conebase2}
\end{align}
where the normalization condition allows us to write the
$r_+$-non-negativity condition as Eq. (\ref{conebase1}) and the
non-negativity of $r_0$ allows us to take the square-root of
Eq. (\ref{sphericalineq}) to obtain Eq. (\ref{gencone}).

For each $\overline{\Psi^{-}}$, we set $r_-$ so as to maximize the left
hand side of Eq. (\ref{gencone}) while satisfying
Eq. (\ref{conebase1}) and Eq. (\ref{conebase2}). Let $d\geq 3$. For
$\overline{\Psi^{-}}\geq0$ we set $r_-=0$, while for
$\overline{\Psi^{-}}\leq0$, we set $r_-=-\overline{\Psi^{-}}$. Considering
these two cases together, we find that the region of joinable
$(\Psi^{-}_{AB}, \Psi^{-}_{BC}, \Psi^{-}_{AC})$ is given precisely by
Eq. (\ref{d1app}) [Eq. (\ref{d1}) in the main text].
If $d=2$, we have $r_-=0$, thus simplifying Eq. (\ref{gencone}) and
Eq. (\ref{conebase1}) to Eq. (\ref{d2bapp}) [Eq. (\ref{d2b})], as desired.
\end{proof}

\begin{thmij} 
\label{trioisoapp}
Two isotropic qubit states $\rho_{AB}$ and $\rho_{AC}$ and qudit Werner state $\rho_{BC}$ with parameters $\Phi^+_{AB}, \Phi^+_{AC}, \Psi^-_{BC}$ are joinable if and only if $(\Phi^+_{AB}, \Phi^+_{AC}, \Psi^-_{BC})$ lies within the cone described by
\begin{eqnarray}
\Phi^{+}_{AB}+\Phi^{+}_{AC}-\Psi^{-}_{BC} \leq d\,  ,
\label{d1isoapp1} 
\end{eqnarray}
 \vspace*{-5mm}
\begin{eqnarray}
 1+\Phi^+_{AB} & + & \Phi^+_{AC} -  \Psi^-_{BC} \geq 
  \label{d1isoapp2}  \\
 \bigg| d(\Psi^-_{BC} & - &1) +  \sqrt{\frac{2d}{d-1}}(e^{i\theta}\Phi^+_{AB}+
e^{-i\theta}\Phi^+_{AC})\bigg|, \nonumber 
\end{eqnarray}
 \vspace*{-5mm}
$$ e^{\pm i\theta}= \pm i \sqrt{(d+1)/(2d)}+\sqrt{(d-1)/(2d)}, $$
or, for $d\geq 3$, within the convex hull of the above cone and the point $(0,0,-1)$.
\end{thmij}

\begin{proof}
By proceeding in analogy to the Werner's case, we need to re-express the non-negativity 
constraints in terms of the relevant reduced state parameters. We have
\begin{align*}
s_1^2+s_2^2&=\tr{}{wS_1}^2+\tr{}{wS_2}^2\nonumber\\
&=|\trb{}{w(S_1+iS_2)}|^2
\end{align*}
and 
\begin{align*}
&S_1+iS_2=\frac{1}{d^2-1}\Big[d\left(V_{(ACB)}^{T_A}+V_{(ABC)}^{T_A}\right)\nonumber\\
&+\left(i\sqrt{d^2-1}-1\right)V_{(AB)}^{T_A} +\left(-i\sqrt{d^2-1}-1\right)V_{(AC)}^{T_A}\Big].
\end{align*}
Setting $s_-=0$ and enforcing the normalization condition $s_+ + s_0=1$, we can 
additionally write
\begin{align}
&\frac{d}{d^2-1}(V_{(ACB)}^{T_A}+V_{(ABC)}^{T_A})\nonumber\\
&=\frac{d}{d+1}(V_{(BC)}-\identity) +\frac{d}{d^2-1}(V_{(AB)}^{T_A}+V_{(AC)}^{T_A}).
\label{eq:normalization}
\end{align}
Thus, the desired expression for $S_1+iS_2$ is
\begin{align*}
S_1+iS_2=\frac{1}{d+1}\Big[d(V_{(BC)}-\identity)&+(i\sqrt{\frac{d+1}{d-1}}+1)V_{(AB)}^{T_A} \nonumber\\
& +(-i\sqrt{\frac{d+1}{d-1}}+1)V_{(AC)}^{T_A}\Big],
\end{align*}
which allows us to write $s_1^2+s_2^2$ in terms of the reduced state parameters:
\begin{equation*}
s_1^2+s_2^2=\Big|d(\Psi^-_{BC}-1)+\sqrt{\frac{2d}{d-1}}(e^{i\theta}\Phi^+_{AB}+
e^{-i\theta}\Phi^+_{AC})\Big|^2 ,
\end{equation*}
where $e^{\pm i\theta}$ is given by the expression above. 

Next, we obtain an expression for $s_0$ in terms of the reduced state parameters. Using 
Eq. (\ref{eq:normalization}), we can write
$$ S_0=\frac{1}{d+1}(\identity-V_{(BC)}+V_{(AB)}^{T_A}+V_{(AC)}^{T_A}).$$
Thus, the spherical non-negativity constraint may be written in terms of the reduced state parameters as the desired result. Furthermore, the non-negativity of $s_0$ provides the boundary forming the base of the cone: $\Phi^{+}_{AB}+\Phi^{+}_{AC}-\Psi^{-}_{BC} \leq d$. If $d\geq3$, we have the possibility that $\tr{}{wS_-}\neq0$. Considering the maximal value of $s_-=1$, we add the point $(0,0,-1)$ as another extremal non-negative point. For $d\geq3$, the convex hull of the cone and this point constitutes the region of reduced state parameter trios which correspond to joinable bipartite states.
\end{proof}

\begin{corowj}
\label{projectioncorowerapp}
Two pairs of qudit Werner states with parameters $\Psi^{-}_{AB}$ and
$\Psi^{-}_{AC}$ are joinable if and only if
$\Psi^{-}_{AB},\Psi^{-}_{AC}\geq-\frac{1}{2}$, or if the parameters satisfy
\begin{equation}
\label{eq:ellipseapp}
 (\Psi^{-}_{AB}+\Psi^{-}_{AC})^2+\frac{1}{3}(\Psi^{-}_{AB}-\Psi^{-}_{AC})^2 \leq 1,
\end{equation}
or additionally, in the case $d\geq3$, if
$\Psi^{-}_{AB},\Psi^{-}_{AC}\leq\frac{1}{2}$.
\end{corowj}
\begin{proof}

To obtain these conditions, it suffices to project the shape given in
Eq. (\ref{d2bapp}) onto the $\Psi^{-}_{AB}$-$\Psi^{-}_{AC}$ plane. The rim of the
cone/bicone projects down to an ellipse whose equation we obtain by
extremizing Eq. (\ref{d1app}) evaluated at the cone base
$\overline{\Psi^{-}}=0$. Setting $\Psi^{-}_{BC}=-\Psi^{-}_{AB}-\Psi^{-}_{AC}$, we find
the boundary to be
$1=(\Psi^{-}_{AB}+\Psi^{-}_{AC})^2+\frac{1}{3}(\Psi^{-}_{AB}-\Psi^{-}_{AC})^2$. Any
pairs of Werner states within this ellipse are joinable. We also have
that $\Psi^{-}_{AB}=1$ and $\Psi^{-}_{AC}=1$ are joinable since setting
$\Psi^{-}_{BC}=1$ causes the three to satisfy Eq. (\ref{d1app}). Then, by
the convexity of the set of joining states, 
the convex hull of $(\Psi^{-}_{AB},\Psi^{-}_{AC})=(1,1)$ and the ellipse corresponds to pairs of
joinable states.
If $d\geq3$, the states $\Psi^{-}_{AB}=-1$ and $\Psi^{-}_{AC}=-1$ are joinable by setting $\Psi^{-}_{BC}=-1$, and hence the joinable $(\Psi^{-}_{AB},\Psi^{-}_{AC}$) pairs also include the convex hull of the ellipse with the point $(\Psi^{-}_{AB},\Psi^{-}_{AC})=(-1,-1)$. 

It remains to show that if $\Psi^{-}_{AB}$ or $\Psi^{-}_{AC}\leq -\frac{1}{2}$ (additionally, $\Psi^{-}_{AB}$ or $\Psi^{-}_{AC} \geq \frac{1}{2}$ for $d\geq3$), then $(\Psi^{-}_{AB},\Psi^{-}_{AC})$ pairs outside of the ellipse are not joinable. To achieve this, we consider a cone viewed from an arbitrary direction an infinite distance away. The shape seen is the shape of the projection. From this vantage point, the circular base of the cone appears as an ellipse. The remaining visible area (seen only if the vertex does not overlap with the base) constitutes the projection of the cone's lateral surface. The boundary of this projection is defined by the two lines extending from the vertex that are tangent to the ellipse. The area contained between these two lines along with the hull of the ellipse constitutes the shape visible from the infinity perspective, or, in other words, the cone's projection. In our case, the points at which these two lines (four lines for the bicone) intersect the ellipse are $(-1/2,1)
$ and $(1,-1/2)$ (additionally, $(-1,1/2)$ and $(1/2,-1)$ for the bicone). Beyond these points, the ellipse ``takes over'' as the projection boundary delimiter. Thus, 
points 
satisfying $\Psi^{-}_{AB}$ or $\Psi^{-}_{AC}\leq -\frac{1}{2}$ (and, $\Psi^{-}_{AB}$ or $\Psi^{-}_{AC} \geq \frac{1}{2}$ for the bicone) are joinable if and only if they are within the ellipse boundary.
\end{proof}

\begin{coroij}
\label{projectioncoroisoapp}
Two pairs of qudit isotropic states with parameters $\Phi^{+}_{AB}$ and
$\Phi^{+}_{AC}$ are joinable if and only if they lie within the convex hull of the ellipse
\begin{equation}
\label{eq:ellipseisoapp}
 \frac{(\Phi^{+}_{AB}/d+\Phi^{+}_{AC}/d-1)^2}{(1/d^2)} +
 \frac{(\Phi^{+}_{AB}/d-\Phi^{+}_{AC}/d)^2}{ (d^2-1)/d^2 } = 1,
\end{equation}
and the point $(\Phi^{+}_{AB},\Phi^{+}_{AC})=(0,0)$.
\end{coroij}
\begin{proof}
We follow the approach of the proof of Corollary III.3 above, by projecting the cone described by 
Eqs. (\ref{d1isoapp1})-(\ref{d1isoapp2}) down to the $\Phi^{+}_{AB}$-$\Phi^{+}_{AC}$ plane. The rim of the cone projects down to an ellipse whose equation we obtain by extremizing Eq. (\ref{d1isoapp1}) evaluated at the cone base $\Phi^{+}_{AB}+\Phi^{+}_{AC}-\Psi^{-}_{BC}=d$. Setting $\Psi^{-}_{BC}=\Phi^{+}_{AB}+\Phi^{+}_{AC}-d$, we find the boundary to be
$1=\left|(\Phi^{+}_{AB}+\Phi^{+}_{AC}-d)+i(\Phi^{+}_{AB}-\Phi^{+}_{AC})/\sqrt{d^2-1}\right|$. From this we easily obtain the ellipse described by Eq. (\ref{eq:ellipseisoapp}). Any pairs of Werner states within this ellipse are joinable. We also have that $\Phi^{+}_{AB}=0$ and $\Phi^{+}_{AC}=0$ are joinable since setting $\Psi^{-}_{BC}=1$ causes the three to satisfy Eq. (\ref{d1isoapp1}). The convex hull of the ellipse and the point $(0,0)$ exhausts the set of joinable pairs of isotropic states.
\end{proof}

Lastly, we explicitly provide the Werner-isotropic hybrid 1-2 joining boundary:

\begin{corowi}
\label{projectioncorohybridapp}
An isotropic state with parameter $\Phi^{+}_{AB}$ and a Werner state with parameter 
$\Psi^{-}_{BC}$ are joinable if and only if they lie within the convex hull of the ellipse
\begin{equation}
\label{eq:ellipsehybridapp}
 \frac{(\Phi^{+}_{AB}/d+\Phi^{+}_{AC}/d-1)^2}{(1/d^2)}+
 \frac{(\Phi^{+}_{AB}/d-\Phi^{+}_{AC}/d)^2}{(d^2-1)/d^2} = 1,
\end{equation}
and the point $(\Phi^{+}_{AB},\Psi^{-}_{BC})=(0,1)$, and, for $d\geq 3$, within the additional convex hull introduced by the point $(\Phi^{+}_{AB},\Psi^{-}_{BC})=(0,1)$.
\end{corowi}
\begin{proof}
Obtained in an analogous manner to Corollaries III.3 and III.4.
\end{proof}

\section{Proof of Werner state $1$-$n$ sharability and construction of $m$-$n$ sharability}
\label{sec:sharingreptools}

Here, we set up the necessary background, then determine the 1-$n$ sharability of 
Werner states, while providing a constructive approach to determing $m$-$n$ sharability.
For fixed $m$, $n$, we only need to determine the most-entangled Werner state 
(largest value of $-\Psi^{-}$) satisfying that sharability property; as noted in
Sec. \ref{sec:S}, all mixtures of that state with any separable state
will necessarily satisfy the sharability property. Thus, in the
one-dimensional convex set parameterized by $\Psi^{-}$, the most-entangled
Werner state that satisfies a sharability property indicates the
boundary between the satisfying and the failing region.

The next step is to map the problem of determining the most-entangled Werner state for a given sharability criterion to the problem of determining the maximal eigenvalue of a particular operator. Specifically, we show that the most-entangled $m$-$n$ sharable Werner state has concurrence equal to the largest eigenvalue of the operator
\begin{equation*}
 H_{m,n}=\frac{1}{mn}\sum_{i\in L} \sum_{j\in R} (-V_{ij}),
\end{equation*}
and that a valid $m$-$n$ sharing state $w_{m,n}$ is given by the normalized projector into the corresponding eigenspace.

We first justify the construction of $H_{m,n}$. From Eq. (\ref{Wc}), the concurrence of a Werner bipartite reduced state of a composite system state $\rho$ is $\mathcal{C}_{ij}=-\tr{}{V_{ij}\otimes \identity_{\overline{ij}}\rho}$. So, for a state $\rho_{m,n}$, with each $L$-$R$ pair a Werner state, its expectation with respect to $H_{m,n}$ is simply the average concurrence of its $L$-$R$ pair reduced states.

Now, that the normalized projector into an eigenspace of $H_{m,n}$ is a valid sharing state follows from the symmetries of this operator: collective unitary invariance ensures that the reduced states are Werner states, while left-system permutation invariance and right-system permutation invariance ensure the $m$-$n$ sharing property. That the maximal eigenvalue of $H_{m,n}$ is the concurrence of the most-entangled $m$-$n$ sharable Werner state follows from the fact that if any state $w_{m,n}'$ were to share Werner states with larger concurrence, its (the sharing state's) expectation value with respect to $H_{m,n}$ would exceed the  maximum eigenvalue, which is a contradiction.

It follows that, by determining the largest eigenvalue of $H_{m,n}$ for each value of $m$ and $n$, we will have characterized the sharability properties of all Werner states.
To determine the eigenvalues of $H_{m,n}$, we express the latter in terms of 
quadratic Casimir operators on $N$-fold tensor product of the (defining) 
$d$-dimensional representations of $\mathfrak{su}(d)$, 
namely, operators of the form 
\begin{align*}
\Lambda_N^2 & \equiv \vec{\Lambda}_{N}\cdot\vec{\Lambda}_{N}
= \sum_{\alpha=1}^{d^2-1}  \left( \Lambda_N^\alpha \right)^2 ,
\end{align*}
where $ \Lambda_N^\alpha = \sum_{i=1}^N (\identity_1 \otimes\ldots \otimes \lambda^{\alpha}_i 
\otimes\ldots\otimes\identity_N)$ and $\{ \lambda^{\alpha}_i \}$ span  each of the $N$ local 
$\mathfrak{su}(d)$ Lie algebras. By construction, any such Casimir operator 
commutes with all the elements of the (generally reducible) tensor-product algebra. 
By rewriting each swap operator as 
$V_{ij}=\identity/d+\sum_{\alpha}\lambda^{\alpha}_i\lambda^{\alpha}_j$, we obtain
\begin{align}
H_{m,n}=&-\frac{\identity}{d}-\frac{1}{mn}\sum_{i\in L}\sum_{j\in
R}\sum_{\alpha}\lambda^{\alpha}_i\lambda^{\alpha}_j  \label{eq:casimir}\\
=&-\frac{\identity}{d}-\frac{1}{2mn}\bigg \{ \sum_{\alpha}\Big[\sum_{k\in L\cup
R}\lambda_k^{\alpha}\Big]^2
-\sum_{\alpha}\Big[\sum_{i\in L}\lambda_i^{\alpha}\Big]^2
\nonumber \\ 
-&\sum_{\alpha}\Big[\sum_{j\in R}\lambda_j^{\alpha}\Big]^2\bigg\}
=\frac{1}{2mn}\left(\Lambda_{L}^2+\Lambda_{R}^2-\Lambda_{LR}^2\right)-
\frac{\identity}{d}. \nonumber 
\end{align}
Two relevant features of tensor-product Casimir operators are worth noting. 
First, $\Lambda_N^2$ will not simply be proportional to the
identity operator, but rather, on each irreducible subspace 
it will act as a (possibly) different multiple of identity. Secondly, the operator
$\Lambda^2_N\otimes\identity_M$ commutes with $\Lambda^2_{N+M}$ for
any $M$ and $N$. 

The three Casimir operators $\Lambda_{LR}^2$, $\Lambda_{L}^2$, and 
$\Lambda_{R}^2$ mutually commute and it is thus meaningful to seek their
simultaneous eigenvalues. Each eigenvalue of a tensor product Casimir
operator corresponds to an irreducible subspace $W_i$ and hence to a
Young diagram. The latter may be used to compute
the value of the corresponding eigenvalue. Following \cite{Fuchs1997},
given a Young diagram $Y$ of column heights $\{a_i\}$ and row lengths
$\{b_j\}$, the eigenvalue of the corresponding space of the Casimir
operator is
\begin{equation}
C_Y=N\left(d-\frac{N}{d}\right)+\sum_j b_j^2-\sum_ia_i^2.
\label{eigenvalue}
\end{equation}
The eigenspaces of $\Lambda_{LR}^2$ which intersect with a given
$\Lambda_L^2$ eigenspace and a given $\Lambda_R^2$ eigenspace may be
calculated by pasting together the boxes of a $Y_L$ and a $Y_R$ Young
diagram in a way that does not cause two previously symmetrized boxes
(same row) to then be antisymmetrized (same column), and vice
versa. 
For example, given the Young diagrams
\begin{equation*}
\yng(2,1)\,\, \text{and}\,\, \yng(2,1),
\end{equation*}
we can construct the following composite diagrams
\begin{align*}
\yng(2,1)\otimes\yng(2,1)=\,\,&\young(\hfil\hfil aa,\hfil
b)\oplus\young(\hfil\hfil aa,\hfil,b)\\ &\oplus\young(\hfil\hfil
a,\hfil ab)\oplus\young(\hfil\hfil a,\hfil a,b)\\
&\oplus\young(\hfil\hfil a,\hfil b,a)\oplus\young(\hfil\hfil,\hfil
a,ab).
\end{align*}
Writing $A_Y=\sum_ia_i^2$ and $B_Y=\sum_j b_j^2$, and replacing each
Casimir operator in Eq. (\ref{eq:casimir}) with its its eigenvalue given
by Eq. (\ref{eigenvalue}), we find that the sum of the first terms of the
eigenvalues cancels with $\identity/d$, leaving 
\begin{equation}
\label{eq:cassumeig}
\text{eig}(H_{m,n})=\frac{A_{Y_{LR}}-A_{Y_{L}}-
A_{Y_{R}}-B_{Y_{LR}}+B_{Y_{L}}+B_{Y_{R}}}{2mn}.
\end{equation}
Thus, we have a prescription for calculating the eigenvalues of
$H_{m,n}$. In particular, we obtain the maximal eigenvalue of $H_{m,n}$ 
by constructing the optimal set of three Young diagrams, as exploited in 
the following:

\begin{thm1n}
\label{thm:1nsharabilityapp}
A qudit Werner state with parameter $\Psi^{-}$ is $1$-$n$ sharable if and
only if
\begin{equation}
\Psi^{-} \geq -\frac{d-1}{n}.
\label{Wsharingapp} 
\end{equation}
\end{thm1n}

\begin{proof}  
All that is necessary is to use Eq. (\ref{eq:cassumeig}) to obtain the largest eigenvalue of $H_{1,n}$. The maximal eigenvalue of $H_{1,n}$ is realized by the following gluing of the 
$1$ and $n$-box Young diagrams
\ytableausetup{mathmode,boxsize=1.2em}
\begin{center}
\begin{ytableau}
*(black)
\end{ytableau}
$\otimes$
\begin{ytableau}
\, &  & \none[\dots] &  \\
\\
\none[\vdots] \\
\\
\end{ytableau}
\,\,\,
$\Rightarrow$\,\,\,
\begin{ytableau}
\, &  & \none[\dots] &  \\
\\
\none[\vdots] \\
\\
*(black).
\end{ytableau}
\end{center}
The values of the $A$ and $B$ here are $A_{Y_{LR}}=d^2+n+1-d$,
$A_{Y_{L}}=1$, $A_{Y_{R}}=(d-1)^2+n+1-d$, and
$B_{Y_{LR}}=B_{Y_{L}}+B_{Y_{R}}=(n+2-d)^2+d-1$, which allow us to
compute the maximal eigenvalue of $H_{1,n}$,
\begin{equation}
\text{max}(H_{1,n})=\frac{d-1}{n}.
\end{equation}
Therefore, the desired conclusion follows. 
\end{proof}
\end{appendix}

\bibliographystyle{apsrev}

\end{document}